\documentclass[sigconf,screen,nonacm]{acmart}
\usepackage[utf8]{inputenc}
\usepackage{mathpartir}
\usepackage{bussproofs}
\usepackage{arydshln}

\usepackage[framemethod=tikz]{mdframed}
\newcommand{\jspace}{\vspace{-2mm}}

\newif\ifarxiv
\arxivfalse
\usepackage{color}

\usepackage[utf8]{inputenc}
\usepackage{subfigure}
\usepackage[english]{babel}
\usepackage{amsmath,amsfonts}
\usepackage{proof}
\usepackage{graphicx}
\usepackage{multicol}
\usepackage{listings}
\usepackage{url}
\usepackage{cmll}
\usepackage{algpseudocode}
\usepackage{algorithmicx}
\usepackage{wrapfig}
\usepackage{enumitem}
\usepackage{textcomp}
\usepackage{latexsym}
\usepackage{epsfig} 
\usepackage{xcolor}
\usepackage{bbm}
\usepackage{etoolbox}
\usepackage{stmaryrd}
\usepackage{xspace}
\usepackage{amsthm}
\usepackage{cleveref}
\newcommand{\vasco}{\ensuremath{\pi_\mathsf{S}}\xspace}

\newcommand{\contcomp}{\circ} % context composition
\newcommand{\emptycont}{\ensuremath{\emptyset}}
\newcommand{\contun}[1]{\m{un} ( #1 ) }
\newcommand{\contlin}[1]{\m{lin} ( #1 )}
\newcommand{\Rec}[2]{\mu #1. #2}%%\newcommand{\Rec}[2][X]{\mu #1. #2}
\newcommand{\recur}[1]{ \ast\,#1  }
\newcommand{\dill}{ \vdash_{\textcolor{darkgray}{\ell}}}

\newcommand{\outsubj}[1]{ \m{os} ( #1 )  }
\newcommand{\nilT}{\ensuremath{{\bf end}}}
\newcommand{\st}{ \vdash_{\textcolor{darkgray}{\mathsf{s}}}}
\newcommand{\lt}{ \vdash_{\textcolor{darkgray}{\mathsf{w}}}}

\newcommand{\rlevel}[2]{ \ast \#^{#1} #2 }
\newcommand{\lts}[1]{ \xrightarrow[]{#1} }
\newcommand{\bn}[1]{\mathsf{bn}(#1)}   % bound names
\newcommand{\fn}[1]{\mathsf{fn}(#1)}   % free names
\newcommand{\names}[1]{\mathsf{n}({#1})}		   %names
\newcommand{\myl}{\ensuremath{l}}
\newcommand{\levelof}[1]{\myl({#1})}		   %level(_)
\newcommand{\lvlass}{\ensuremath{\diamond}}
\newcommand{\weight}[1]{\mathsf{wt}({#1})}		   %wt(_)
\newcommand{\wtless}{ \prec }            % wt(_) < wt(_)
\newcommand{\zerovec}{ \ensuremath{\mathtt{0}} }            % wt(_) < wt(_)

\newcommand{\dillcontrel}[2]{\ensuremath{\asymp^{\,#1}_{\,#2}}}
\newcommand{\vaslinunsplit}{\ensuremath{\circledast}}
\newcommand{\pidill}{\ensuremath{\pi_{\sf DILL}}\xspace}
\newcommand{\pilvl}{\ensuremath{\pi_{\sf W}}\xspace}
\newcommand{\vasout}[3]{\ensuremath{  \ov #1 \out #2 . #3  }}
\newcommand{\vasin}[4]{\ensuremath{   #1 \  #2\inp #3.#4  }}
\newcommand{\vaspara}[2]{\ensuremath{ #1 \pp #2 }}
\newcommand{\vasres}[2]{\ensuremath{  \res {#1} #2  }}

\newcommand{\vassend}[2]{\ensuremath{  \oc #1 . #2  }}
\newcommand{\vasreci}[2]{\ensuremath{  \wn #1 . #2  }}
\newcommand{\vasred}{\longrightarrow}
\newcommand{\vaslang}{\ensuremath{\mathcal{S}}\xspace}

\newcommand{\fv}[1]{\mathsf{fv}(#1)}   % free var
\newcommand{\bv}[1]{\mathsf{bv}(#1)}   % free var

\newcommand{\sep}{\ | \ } 

\newcommand{\vastype}[1]{\ensuremath{\mathtt{[S{:}#1]} }}

\makeatletter
\def\@tempa#1{\@xp\@tempb\meaning#1\@nil#1}
\def\@tempb#1>#2#3 #4\@nil#5{%
  \@xp\ifx\csname#3\endcsname\mathaccent
    \@tempc#4?"7777\@nil#5%
  \else
    \PackageWarningNoLine{amsmath}{%
      Unable to redefine math accent \string#5}%
  \fi
}
\def\@tempc#1"#2#3#4#5#6\@nil#7{%
  \chardef\@tempd="#3\relax\set@mathaccent\@tempd{#7}{#2}{#4#5}}

\@tempa\widehat

\newcommand{\lvlint}[3]{\ensuremath{ \#^{#1} ( #2 , #3) }}
\newcommand{\lvloutt}[3]{\ensuremath{ \#^{#1} \langle #2 , #3 \rangle  }}
\newcommand{\lvlunst}[2]{\ensuremath{ \ast \#^{#1} (#2)  }}

\newcommand{\lvlunct}[2]{\ensuremath{ \ast \#^{ #1  } \langle #2 \rangle  }}
\newcommand{\lvlinpoly}[4]{\ensuremath{  #1 \inp {\tilde{#2}}. #4  }}
\newcommand{\lvloutpoly}[4]{\ensuremath{  \ov #1 \out {\tilde{#2} }.#4  }}
\newcommand{\lvlservpoly}[4]{\ensuremath{  \bang #1 \inp {\tilde{#2} }. #4  }}
\newcommand{\lvlin}[4]{\ensuremath{  #1 \inp {{#2}, #3}. #4  }}
\newcommand{\lvlout}[4]{\ensuremath{  \ov #1 \out {{#2}, #3 }.#4  }}
\newcommand{\lvlserv}[4]{\ensuremath{  \bang #1 \inp {{#2}, #3 }. #4  }}
\newcommand{\lvlsplit}[2]{\ensuremath{ #1 :: #2   }}
\newcommand{\dualjoin}[2]{\ensuremath{ \langle #1 , #2 \rangle    }}

\newcommand{\lvlnilT}{\ensuremath{{\bf unit}  }}
\newcommand{\lvllang}{\ensuremath{\mathcal{W}}\xspace}
\newcommand{\lvltype}[1]{\ensuremath{\mathtt{[W{:}#1]} }}
\newcommand{\dillint}[2]{\ensuremath{ #1  \lolli #2 }}
\newcommand{\dilloutt}[2]{\ensuremath{ #1 \tensor #2  }}
\newcommand{\dillunt}[1]{\ensuremath{  \bang #1 }}
\newcommand{\dillchoicet}[2]{\ensuremath{ #1 \oplus #2  }}
\newcommand{\dillselet}[2]{\ensuremath{ #1  \&   #2  }}
\newcommand{\dillnilT}{\ensuremath{ \one }}
\newcommand{\dillred}{\longrightarrow}
\newcommand{\dillin}[3]{\ensuremath{  #1 \inp {#2 }. #3  }}
\newcommand{\dillbout}[3]{\ensuremath{  \ov #1 (#2).#3 }}
\newcommand{\dillout}[3]{\ensuremath{  #1 \out { #2 }.#3 }}
\newcommand{\dillserv}[3]{\ensuremath{  \bang #1 \inp {#2  }. #3  }}
\newcommand{\dillchoice}[3]{\ensuremath{ #1. \mathbf{case} (#2 , #3)  }}
\newcommand{\dillselel}[2]{\ensuremath{ #1. \mathbf{inl} ;#2   }}
\newcommand{\dillseler}[2]{\ensuremath{ #1. \mathbf{inr} ;#2  }}
\newcommand{\dillforward}[2]{\ensuremath{ [ #1 \leftrightarrow #2 ] }}
\newcommand{\dillfwdbang}[2]{\ensuremath{ \bang [ #1 \leftrightarrow #2 ] }}
\newcommand{\dilllang}{\ensuremath{\mathcal{L}}\xspace}
\newcommand{\dillnotun}[1]{\ensuremath{{#1}^{\textcolor{red}{\dagger}} }}
\newcommand{\dillnotunsingular}[1]{\ensuremath{\dagger(#1)}}
\newcommand{\server}[1]{\ensuremath{\mathsf{svr} ( #1 )  }}
\newcommand{\client}[1]{\ensuremath{\mathsf{cli} ( #1 )  }}
\newcommand{\notserver}[1]{\ensuremath{\neg \server{#1}  }}
\newcommand{\notclient}[1]{\ensuremath{ \neg \client{#1}  }}
\newcommand{\dilltype}[1]{\ensuremath{\mathtt{[L{:}#1]} }}
\newcommand{\ttrue}{\ensuremath{\mathtt{true}}\xspace}
\newcommand{\ffalse}{\ensuremath{\mathtt{false}}\xspace}
\newcommand{\sttoltt}[3]{ \textcolor{blue}{\llparenthesis} #1 \textcolor{blue}{\rrparenthesis}_{#2}^{#3}}
\newcommand{\sttoltp}[2]{\ensuremath{\textcolor{blue}{\langle\!|} #1 \textcolor{blue}{|\!\rangle}^{#2}}}
\newcommand{\sttoltj}[3]{\textcolor{blue}{\left\llbracket \textcolor{black}{#1} \right\rrbracket}_{#2}^{#3}}

\newcommand{\sttodillt}[1]{ \textcolor{red}{\llparenthesis} #1 \textcolor{red}{\rrparenthesis}}
\newcommand{\sttodillp}[1]{\ensuremath{\textcolor{red}{\langle\!|} #1 \textcolor{red}{|\!\rangle}}}
\newcommand{\sttodillj}[1]{\textcolor{red}{\left\llbracket \textcolor{black}{#1} \right\rrbracket}}

\newcommand{\abconditionO}[6]{\ensuremath{\dagger}}
\newcommand{\abconditionT}[7]{\ensuremath{\ddagger}}
\newcommand{\abconditionTstar}[9]{\ensuremath{\ddagger^\star}}

\newcommand{\m}[1]{\mathsf{#1}}

\newcommand{\bang}{\boldsymbol{!}}

\newcommand{\nil}{{\mathbf{0}}}

\newcommand{\out}[1]{\langle #1\rangle}
\newcommand{\inp}[1]{(#1)}
\newcommand{\res}[1]{({ \nu} #1)}
\newcommand{\pp}{\ {|}\ }
\newcommand{\un}{\m{un}}
\newcommand{\lin}{\m{lin}}

\newcommand{\Didtext}[1]{{\scriptsize{\textsc{#1}}}}
\newcommand{\Did}[1]{({\Didtext{#1}})}
\newcommand{\dom}[1]{\m{dom}(#1)}

\newcommand{\lolli}{\mathord{\multimap}}
\newcommand{\tensor}{\otimes}

\newcommand{\one}{\ensuremath{{\bf{ 1}}}}

\newcommand{\dual}[1]{\overline{#1}}

\newcommand{\ov}[1]{\overline{#1}}

\def\substj#1#2{[\raisebox{.5ex}{\small$#1$}\! / \mbox{\small$#2$}]}

\newcommand{\tra}[1]{\xrightarrow{#1}}
\newcommand{\redd}{\tra{~~~}}

\crefname{enumi}{Part}{Parts}
\crefname{section}{\S\!\!\,}{\S\!\!\,}%
\crefname{subsection}{\S\!\!\,}{\S\!\!\,}%
\crefname{subsubsection}{\S\!\!\,}{\S\!\!\,}%
\crefname{appendix}{\S\!\!\,}{\S\!\!\,}

\newtheorem{theorem}{Theorem}[section]

\newtheorem{definition}{Definition}[section]
\newtheorem{notation}{Notation}[section]
\newtheorem{remark}{Remark}
\newtheorem{example}{Example}[section]
\newtheorem{proposition}{Proposition}[section]
\newtheorem{corollary}{Corollary}[section]

\copyrightyear{2023} 
\acmYear{2023} 
\setcopyright{acmlicensed}\acmConference[PPDP 2023]{International Symposium on Principles and Practice of Declarative Programming}{October 22--23, 2023}{Lisboa, Portugal}
\acmBooktitle{International Symposium on Principles and Practice of Declarative Programming (PPDP 2023), October 22--23, 2023, Lisboa, Portugal}
\acmPrice{15.00}
\acmDOI{10.1145/3610612.3610615}
\acmISBN{979-8-4007-0812-1/23/10}

\begin{document}

\title{Termination in Concurrency, Revisited}


\author{Joseph W. N. Paulus}
\orcid{0000-0002-1711-9361}
\affiliation{%
  \institution{University of Groningen}
  \streetaddress{1 Th{\o}rv{\"a}ld Circle}
  \city{Groningen}
  \country{The Netherlands}}
%\email{email@affiliation.org}

\author{Jorge A. Pérez}
\orcid{0000-0002-1452-6180}
\affiliation{%
  \institution{University of Groningen}
  \streetaddress{1 Th{\o}rv{\"a}ld Circle}
  \city{Groningen}
  \country{The Netherlands}}
%\email{email@affiliation.org}

\author{Daniele Nantes-Sobrinho}
\orcid{0000-0002-1959-8730}
\affiliation{%
  \institution{Imperial College London}
  \streetaddress{1 Th{\o}rv{\"a}ld Circle}
  \city{London}
  \country{UK}
  }
 \additionalaffiliation{% 
  \institution{University of Bras\'{i}lia, Brazil}
  \city{Bras\'{i}lia}
  \country{Brazil}
  }
%\email{dnantess@imperial.ac.uk}

\renewcommand{\shortauthors}{Paulus et al.}

\begin{abstract}
Termination is a central property in sequential programming models: a term is terminating if all its reduction sequences are finite. 
Termination is also important in concurrency in general, and for message-passing programs in particular.
A variety of type systems that enforce termination by typing have been developed. %in the last 20 years. 
%Unfortunately, %the precise relation between those type systems remains poorly understood; in particular, 
%we still understand little about the relative  strengths of  different static mechanisms for enforcing termination. 
In this paper, we rigorously compare several type systems for $\pi$-calculus processes from the unifying perspective of termination.
Adopting \emph{session types} as reference framework, we consider two different type systems: one follows Deng and Sangiorgi's weight-based approach; the other is Caires and Pfenning's Curry-Howard correspondence between linear logic and session types. 
Our technical results precisely connect these very different type systems, and shed light on the classes of client/server interactions they admit as correct. 
\end{abstract}





\maketitle

\section{Introduction}
The purpose of this paper is to present the first comparative study of type systems that enforce \emph{termination} for message-passing processes in the $\pi$-calculus, the paradigmatic model of concurrency.
%Our study involves a number of technical challenges, as we explain next. 

Termination is a cornerstone of sequential programming models: a term is terminating if all its reduction sequences are finite. 
Termination is also an important property in concurrency in general, and in message-passing programs in particular. 
In such a setting, infinite sequences of internal steps are rather undesirable, as they could jeopardize the reliable interaction between a process and its environment. 
That is, we would like processes that exhibit \emph{infinite} sequences of observable actions, possibly intertwined with \emph{finite} sequences of internal/unobservable steps (i.e., reductions). 

In the (un)typed  $\pi$-calculus, infinite behavior can be  expressed via operators for recursion (or recursive definitions) or replication. 
We are interested in replication, and in particular in \emph{input-guarded} replication, denoted $!x(y).P$.
Input-guarded replication neatly captures the essence of \emph{servers} that are persistently available to spawn interactive behavior upon invocations by concurrent \emph{clients}.
This way, it   precisely expresses the controlled invocation of (shared) resources. 
To understand its operation, let us write $x\langle z\rangle$ to denote an output prefix, intended as an invocation to a server such as $!x(y).P$.
The corresponding reduction rule is then roughly as follows:
$$!x(y).P \pp x\langle z\rangle.Q ~\longrightarrow~ !x(y).P \pp P\substj{z}{y} \pp Q $$
Thus, after a synchronization on $x$, the server $!x(y).P$ continues to be available, and a copy of $P$ is spawned (where $\substj{z}{y}$ denotes the substitution of $y$ with $z$, as usual), enabling interaction with $Q$.

In this setting, an obvious source of non-terminating behaviors is when clients and servers invoke each other indefinitely. 
This situation arises, in particular,  when client invocations occur in the body of a server, which can easily trigger infinite ``ping-pong'' reductions, as in the following process (where $\nil$ denotes inaction):
\begin{equation}
!x(y).x\langle y\rangle.\nil \pp x\langle w\rangle.\nil ~\longrightarrow~ !x(y).x\langle y\rangle.\nil \pp x\langle w\rangle.\nil \pp \nil ~\longrightarrow~ \cdots 
\label{eq:ping}	
\end{equation}
The challenge of statically ruling out processes such as \eqref{eq:ping} while enabling expressive client/server interactions has been addressed by multiple authors via various type systems, see, e.g.,~\cite{DBLP:conf/lics/YoshidaBH01,DBLP:conf/ifipTCS/DengS04,DBLP:journals/mscs/Sangiorgi06,DBLP:conf/concur/DemangeonHS10,DBLP:journals/toplas/KobayashiS10,DBLP:journals/fuin/Piccolo12,DBLP:conf/tgc/ToninhoCP14,DBLP:journals/pacmpl/LagoVMY19}. 
Their underlying approaches are vastly diverse.
For instance, 
  Yoshida et al.~\cite{DBLP:conf/lics/YoshidaBH01} adopt a type-theoretical approach  based on logical relations and linear action types.
 Deng and Sangiorgi~\cite{DBLP:conf/ifipTCS/DengS04} transport ideas from rewriting systems (well-founded measures) into a $\pi$-calculus with simple types.
Caires and Pfenning's Curry-Howard correspondence between linear logic and session types~\cite{CairesP10} represents yet another approach: their type system enforces termination based purely on proof-theoretical principles, by  interpreting the exponential `$!A$' as the type of a server and by connecting cut elimination with process synchronization.
Several natural questions arise. How do these type disciplines compare? What are their relative strengths? More concretely, are there terminating processes detected as such by one type system but not by some other? If so, where is the difference?

As inviting and intriguing these questions are, a technical approach to a formal comparison is far from obvious. 
An immediate obstacle concerns the underlying formal models: all the type systems mentioned above operate on \emph{different dialects} of the $\pi$-calculus, involving, e.g., synchronous/asynchronous communication, and monadic/polyadic message passing. These differences quickly escalate at the level of the respective type systems, with the presence/absence of \emph{linearity} unsurprisingly playing a key distinguishing role. How do we even start formulating the intended comparison?

We frame our formal comparison as follows.
As baseline for comparison we take the   $\pi$-calculus processes typable with Vasconcelos's session type system~\cite{V12}.  This is a quite liberal type system, which induces a  broad class of session processes (including non-terminating ones), which is convenient for our purposes. 
In the following, this baseline class of processes is denoted $\vaslang$. 

We then consider two representative classes of processes, both terminating by typing. One is based on Deng and Sangiorgi's \emph{weight-based} type system; the other is Caires and Pfenning's linear-logic type system. Because these type systems are so different from Vasconcelos's, to connect them with \vaslang we require  typed translations. This leads to two classes of terminating processes:
%%\begin{wrapfigure}{t}{0.3\textwidth}
%\begin{figure}[!t]
%\includegraphics[width= .3\textwidth]{PPDP-2.jpg}
%%\end{wrapfigure}
%\end{figure}
\begin{itemize}
	\item  \lvllang contains all processes in \vaslang (i.e., typable under 
	Vasconcelos's type system) which are also typable (up to a translation) under the weight-based type system.
	\item  \dilllang contains all processes in \vaslang  which are also typable (up to another translation) by the Curry-Howard correspondence.
\end{itemize}

 This way, because Vasconcelos's  system can type non-terminating processes, both $\lvllang \subset \vaslang$ and $\dilllang \subset \vaslang$ hold by definition. 
Our technical contributions are two-fold. 
\begin{enumerate} 
	\item 
Because the type systems by Vasconcelos and by Deng and Sangiorgi are so different, to define \lvllang we develop a \emph{new weight-based type system} that combines elements from both: it ensures termination by enforcing well-founded measures (as Deng and Sangiorgi's) while accounting for linearity and sessions (as Vasconcelos's). The translation involved in bridging \vaslang and this new type system determines a technique for ensuring termination of session-typed processes, which is new and of independent interest.
\item We prove that $\dilllang \subset \lvllang$
but 
$\lvllang \not \subset \dilllang$, thus determining the exact relationship between  these classes of typed processes.
Our discovery is that there are terminating session-typed processes that are typable with the weight-based approach but not under the Curry-Howard correspondence. 
In other words, techniques based on well-founded measures turn out to be more powerful for enforcing termination than proof-theoretical foundations.
\end{enumerate}
Next, we introduce the class \vaslang.
\Cref{s:weight} develops the new weight-based type system and \Cref{s:lvllang} studies its corresponding class \lvllang.
The Curry-Howard correspondence for concurrency is recalled in \Cref{s:pas}, and its corresponding class \dilllang is presented in \Cref{s:dilllang}.
Finally, \Cref{s:close} collects concluding remarks.
%\ifarxiv
%The appendix contains omitted technical material.
%\else
%The full version of the paper~\cite{X} contains omitted technical material.
%\fi 
\section{The Class $\vaslang$ of Session Processes}

We present the process language that we shall consider as reference in our comparisons, and its corresponding session type system.
We distinguish between (i)~the processes induced by this process model and (ii)~the class of well-typed processes (\Cref{d:vaslang}); in the following, these classes are denoted by \vasco and $\vaslang$, respectively.
We consider the type system by Vasconcelos~\cite{V12}, which ensures communication safety and session fidelity, but not progress/deadlock-freedom nor termination. 
Our presentation closely follows~\cite{V12}, pointing out differences where appropriate. 

\subsection{The Process Model \vasco}
\label{ss:pm}

\begin{definition}[Processes]
	\label{def:vascosyntax}
	Let  $x, y, \ldots$ range over \emph{variables}, denoting \emph{channel names} (or \emph{session endpoints}), and $v, v', \ldots$ over \emph{values}; 
for simplicity, the sets of values and variables coincide.
%\jpedit{In examples, we use $\unit$ to denote a terminated channel that cannot be further used.}
Also, let  $P, Q, \ldots$ range over \emph{processes}, defined by the grammar of \Cref{f:vascosyntax}, which induces the class \vasco.
\end{definition}

\begin{figure}[t!]
	\begin{mdframed}
	\[
	\begin{aligned}
		P,Q ::= &								&& \mbox{(Processes)}\\
				&  \vasout{x}{v}{P}				&& \mbox{(output)} 
				&&  \vasin{q}{x}{y}{P}			&& \mbox{(input)} \\
				&  \vaspara{P}{Q}  		  		&& \mbox{(composition)} 
				&&  \vasres{xy}{P}				&& \mbox{(restriction)} \\
				&  \nil	     					&& \mbox{(inaction)} \\
		q   ::= &								&& \mbox{(Qualifiers)}\\
				& \lin 							&& \mbox{(linear)}
				&& \un 							&& \mbox{(unrestricted)}\\
		v	::= &								&& \mbox{(Values)}\\
				& x								&& \mbox{(variables)}\\
	\end{aligned}
	\]
	\end{mdframed}
	\jspace
	\caption{Syntax of the session $\pi$-calculus \vasco}
	\label{f:vascosyntax}
	\jspace
\end{figure}

The output process $\vasout{x}{v}{P}$ sends value $v$ across channel $x$ and then continues as $P$. 
In the input process $ \vasin{q}{x}{y}{P}$, the qualifier $q$ can be either $\lin$ (denoting a linear input) or $\un$ (denoting an unrestricted input, i.e., a replicated server). 
In either case, $x$ expects to receive a value that will replace free occurrences of $y$ in $P$. 
Parallel composition $\vaspara{P}{Q}$ denotes the concurrent execution of processes $P$ and $Q$. 
The process $\vasres{xy}{P}$ denotes the restriction of the \emph{co-variables} $x$ and $y$ with scope $P$. This declares them as dual endpoints, which are expected to behave complementarily to each other. 
We write $\res {zv:S} P $ when either $z$ or $v$ have session type $S$ in $P$.
As we will see, a synchronization  always occurs across a pairs of co-variables.
Finally, the inactive process is denoted as $\nil$.

As usual, the set of free variables in a process $P$ is denoted $\fv{P}$, and similarly $\bv{P}$ for bound variables. 
The capture-free substitution of the variable $z$ by the value $v$ is denoted as $\substj{v}{z}$. We adopt Barendregt's variable convention.

With respect to \cite{V12}, the above the process syntax leaves out boolean values, conditional expressions, and labeled choices, which are all inessential for our comparative study of termination. 

\begin{definition}[Reduction Semantics]
The reduction relation $\vasred$ of \vasco is defined in \Cref{f:vascosymantics}.
\end{definition}

\begin{figure}[t!]
	\begin{mdframed}[leftmargin = -3mm, rightmargin = -3mm]
	\[
		\begin{aligned}
			\vaspara{P}{Q} &\equiv \vaspara{Q}{P}
			&
			\vaspara{P}{\nil} & \equiv P
			\\
			 \vaspara{(\vaspara{P}{Q})}{R} &\equiv \vaspara{P}{(\vaspara{Q}{R})} 
			&
			\vasres{xy}{\nil} & \equiv \nil
			\\
			\vasres{xy}\vasres{zw}{P} & \equiv \vasres{zw}\vasres{xy}{P}
			&
			\!\!\vasres{xy}{P} & \equiv \vasres{yx}{P}~ 
			\\
			\vaspara{(\vasres{xy}P)}{Q} & \equiv \vasres{xy}(\vaspara{P}{Q}) \,\text{[$x,y \not \in \fv{Q}$]}
		\end{aligned}
	\]
	\hrule
	\smallskip
	\[
		\begin{array}{rll}  
		\Did{R-LinCom} & 
			\vasres{xy}{(  
					\vaspara{\vasout{x}{v}{P}}
						{\vaspara{\vasin{\lin}{y}{z}{Q}}
							{R}} 
				   )}  
			\\
			& \quad \vasred \vasres{xy}{(
					\vaspara{P}
						{\vaspara{Q\substj{v}{z}}
							{R}}
	
				)}
			\\[1mm] 
		\Did{R-UnCom} & 
			\vasres{xy}{(
					\vaspara{\vasout{x}{v}{P}}
						{\vaspara{\vasin{\un}{y}{z}{Q}}
							{R}}
				)} 
				\\
			& \qquad \vasred
			\vasres{xy}{(
					\vaspara{P}
						{\vaspara{Q\substj{v}{z}}
							{\vaspara{\vasin{\un}{y}{z}{Q}}
								{R}}
						}
				)}
			\\[1mm]  
			\Did{R-Par}&		{P\vasred Q}\Longrightarrow
								{\vaspara{P}{R} \vasred \vaspara{Q}{R}}
		\\[1mm]
		\Did{R-Res}&	{P\vasred Q} \Longrightarrow
								{\vasres{xy}{P} \vasred \vasres{xy}{Q}}
		   \\[1mm]
		\Did{R-Str}& 	{P\equiv P',\ P\vasred Q,\ Q'\equiv Q}
									\Longrightarrow
									{P'\vasred Q'} & 
	%								\vspace{2mm}\\
		\end{array}
	\]
	\end{mdframed}
	\jspace
	\caption{Reduction semantics for \vasco}
	\label{f:vascosymantics}
	\jspace
\end{figure}

The reduction semantics for \vasco follows standard lines for (session) $\pi$-calculi; it is closed under a structural congruence, denoted $\equiv$, which captures expected principles for parallel composition and restriction.
The reduction rule $\Did{R-LinCom}$ captures the linear communication across co-variables $x$ and $y$, appropriately declared by restriction, in which value $v$ is exchanged.
Similarly, rule $\Did{R-UnCom}$ denotes unrestricted communication across co-variables; in this case, the input prefix is persistent, and remains ready for further synchronizations after reduction.
The contextual rules $\Did{R-Par}$ and $\Did{R-Res}$ express that concurrent processes can reduce within the scope of parallel composition and restriction. Finally, rule $\Did{R-Str}$  denotes that reductions are closed under structural congruence.

\subsection{Session Types}\label{ss:typesess}
%\checkthis{JP: We need to explain exactly how this system relates to that in ~\cite{V12}, and move all relevant statements to the main text.}
%
%\checkthis{The most important thing to discuss here is the fact that Vasco uses pretypes and lin/un, instead we only consider linear ones. The only un type is end.}

We endow \vasco with the session type system by Vasconcelos~\cite{V12}, which ensures that well-typed processes respect their protocols but does not ensure deadlock-freedom nor termination guarantees.
With respect to the syntax of types in~\cite{V12},  we only consider channel endpoint types (no ground types such as \textsf{bool}).

\begin{definition}[Session Types]
The syntax of {session types} ($T, S, \ldots$) is given in \Cref{f:vascosessiontypes}.
\end{definition}

\begin{figure}[t!]
	\begin{mdframed}
	\[
%	\begin{aligned}
%		q   ::= &								&& \mbox{(Qualifiers)}\\
%				& \lin 							&& \mbox{(linear)}\\
%				& \un 							&& \mbox{(unrestricted)}\\
%		p   ::= &								&& \mbox{(Pretypes)}\\
%				& \vasreci{T}{S}				&& \mbox{(receive)}\\
%				& \vassend{T}{S}				&& \mbox{(send)}\\
%		T,S ::= &								&& \mbox{(Types)}\\
%				& \nilT 						&& \mbox{(termination)}\\
%				& q \ p 						&& \mbox{(qualified pretypes)}\\
%				& a 							&& \mbox{(type variable)}\\
%				& \Rec{a}{T}					&& \mbox{(recursive types)}\\
%		\Gamma ::= &							&& \mbox{(Contexts)}\\
%				& \emptyset						&& \mbox{(empty context)}\\
%				& \Gamma , x:T					&& \mbox{(assumption)}\\
%	\end{aligned}
	\begin{aligned}
		q   ::= &	&							
		 \mbox{(Qualifiers)}
		 & & 
		 		T,S ::= &	&							 \mbox{(Types)} &&
		\\
				& \lin 							& \mbox{(linear)}
				& & &
				\nilT 						& \mbox{(termination)}
				\\
				& \un 							& \mbox{(unrestricted)}
				& & & q \ p  & \mbox{(pretypes)}
				\\
		p   ::= &								& \mbox{(Pretypes)}
		& & & 
		a & \mbox{(type variable)}
		\\
				& \vasreci{T}{S}				& \mbox{(receive)}
				& & & 
				\Rec{a}{T}	 & \mbox{(recursive types)}
				\\
				& \vassend{T}{S}				& \mbox{(send)}
				& & & 
				 & 
		\\
			\Gamma	 ::= &  						& \mbox{(Contexts)}
				& & & &
				\\
				& \emptyset						& \mbox{(empty)}
				& & & &
				\\
				& \Gamma , x:T
				& \mbox{(assumption)} & 
	\end{aligned}
	\]
	\end{mdframed}
	\jspace
	\caption{Session Types of \vasco}
	\label{f:vascosessiontypes}
	\jspace
\end{figure}

Session types $T, S$ describe protocols as \emph{sequences} of actions for an endpoint; they do not admit the parallel usage of an endpoint.
They have the following forms:
\begin{enumerate}
	\item Type $\nilT$ is given to an endpoint with a completed protocol.
	
\item Type $ q \ p $ denotes pre-type $p$ with qualifier $q$, which indicates either a linear or an unrestricted behavior ($\lin$ and $\un$, respectively). 
The pre-type $\wn T.S$ is given to an endpoint that first receives a value of type $T$ and then continues according to type~$S$.
	Dually, the pre-type $\oc T.S$ is intended for an endpoint that first outputs a value of type $T$ and then continues according to  $S$.
	
	\item Type $\Rec{a}{T}$ is a recursive type, with type variable $a$. A recursive type is required to be \emph{contractive}, {i.e., it contains no subexpression of type $\Rec{a_1}\ldots \Rec{a_n}a_1$}; and $a$ is bound with scope $T$. Notions of bound and free type variables, alpha-conversion and capture-avoiding substitutions (denoted $\substj{S}{a}$) is defined as usual. Type equality is based on regular infinite trees \cite{V12}.
\end{enumerate} 

Recursive types that are \emph{tail-recursive} are expressive enough to type servers and clients; we have a dedicated notation for them.

\begin{notation}[Server and Client Types]\label{note:rectype}
	We shall write $\recur{\wn T}$ to denote the server type $ \Rec{a}{\un\  \wn T.a}$, where variable $a$ does not occur in $T$. 
	Similarly, we write $\recur{\oc T}$ to denote the client type $ \Rec{a}{\un\  \oc T.a}$
\end{notation}

In the following, we shall work  with tail-recursive types only.
A central notion in session-based concurrency is  \emph{duality}, 
which relates session types offering opposite (i.e., complementary) behaviors; it stands at the basis of communication safety and session fidelity. 

\begin{definition}[Duality]
	Given a (tail-recursive) session type $T$, its dual type $\dual{T}$ is defined as follows:
	\begin{displaymath}
		\begin{array}{rclrclrcl}
			\dual{\nilT} &=&  \nilT &\quad  \dual{\oc T.S} &=& \wn T .\dual{S} & \quad \dual{\recur{\wn {T}}} &=&  \recur{\oc {T}}
			\\
			\dual{q \ p} &=& q \ \dual{p }&
			\dual{\wn T.S} &=& \oc T.\dual{S} &
			\dual{ \recur{\oc {T}}} &=& \recur{\wn {T}} 
		\end{array}
		\end{displaymath}
%	\begin{displaymath}
%	\begin{array}{rll}
%		\dualbh{\nilT, \sigma} &\defeq& \nilT\\
%		\dualbh{q \ p, \sigma} &\defeq& q \ \dualbh{p, \sigma }\\
%		\dualbh{\oc T.S, \sigma} &\defeq& \wn (T \ \sigma).\dualbh{S, \sigma}\\
%		\dualbh{\wn T.S, \sigma} &\defeq& \oc (T \ \sigma).\dualbh{S, \sigma}\\
%		\dualbh{a, \sigma} &\defeq& a\\
%		\dualbh{\Rec{a}{T}, \sigma} &\defeq& \Rec{a}{\dualbh{T, \substj{\Rec{a}{T}}{a} ; \sigma}} \
%	\end{array}
%	\end{displaymath}

%	duality of a session type $T$ is given by $\dual{T} = \dualbh{T} =  \dualbh{T, \epsilon}$
\end{definition}

We now collect definitions and results from~\cite{V12} that will lead to state the main properties of typable processes. 

\begin{definition}[Predicates on Types/Contexts]\label{def:pred_cont}
We consider two predicates on types, denoted $\contlin{T}$ and $\contun{T}$, defined as follows:
\begin{itemize}
	\item $\contun{T}$ if and only if $T = \nilT$ or $T = \un \ p$. 
	\item $\contlin{T}$ if and only if $\ttrue$.
\end{itemize}
The definition extends to contexts as follows: we write
$q(\Gamma)$ if and only if $x:T \in \Gamma$ implies $q(T)$.
\end{definition}

This way, to express that $T$ defines strictly linear behavior we write $\neg \contun{T}$ (and similarly for a context $\Gamma$).
The following notation is useful to separate the linear and unrestricted portions of a context:
\begin{notation}\label{note:sep}
We write $\Gamma \vaslinunsplit \Gamma'$ if $\contun{\Gamma} \land \neg \contun{\Gamma'}$.
\end{notation}

%unrestricted types can behave linearly and hence the linear predicate is always true however linear types only have the 

\begin{definition}[Context Split and Update]\label{d:splitsts}
The split and update operations on contexts, denoted $\contcomp$ and $+$, are defined as follows.
	\begin{mathpar}
		\inferrule{}
		{ \emptycont \contcomp \emptycont = \emptycont }
		\and
		\inferrule{ \Gamma_1 \contcomp \Gamma_2 = \Gamma \\ \contun{T}}
		{\Gamma , x:T = ( \Gamma_1 , x:T ) \contcomp ( \Gamma_2 , x: T ) }
		\and
		\inferrule{ \Gamma_1 \contcomp \Gamma_2 = \Gamma }
		{ \Gamma , x: \lin\  p = (\Gamma_1 , x: \lin\  p) \contcomp \Gamma_2 }
		\and
		\inferrule{  \Gamma_1 \contcomp \Gamma_2 = \Gamma }
		{ \Gamma , x: \lin \  p = \Gamma_1 \contcomp ( \Gamma_2 , x: \lin \ p ) }
		\and
		\inferrule{ x:U \not \in \Gamma }
		{ \Gamma + x:T = \Gamma , x:T }
		\and
		\inferrule{ \contun{T} }
		{ ( \Gamma , x:T ) + x:T = ( \Gamma , x:T  ) }
	\end{mathpar}
\end{definition}

The typing system considers two kinds of judgments, for processes and for variables, denoted  $\Gamma \st P$ and $\Gamma\st x:T$, respectively. 
We write $~ \st P$ when $\Gamma$ is empty.
The typing rules are given in \Cref{f:vasorigrules}.
% \jpedit{Need explanations for the rules.}
We will explain Rule \vastype{In}: it  is parametric on the qualifiers $q_1$ and $q_2$  and  covers three different behaviours depending on whether $q_i$ is $\lin$ or $\un$, for $i=1,2$. In the case $q_1=\lin$, to prove $\Gamma_1 \contcomp \Gamma_2 \st\lin\ x\inp y.P$, we need to prove $\Gamma_1 \st x: q_2 \wn T.S $ and $ (\Gamma_2 + x: S), y:T \st P$; note  that $\lin (\Gamma_1\contcomp \Gamma_2)$ is true, by \Cref{def:pred_cont}.   In the case $q_2=\lin$, both judgments hold if $\Gamma_1=\Gamma_1',  x: \lin \wn T.S$, the assignment $x: \lin \wn T.S$ does not occur in $\Gamma_2$, by \Cref{d:splitsts},  and $x:S$ is added to $\Gamma_2$ for the continuation.   
Differently, when $q_2=\un$, both judgments hold if $\Gamma_1=\Gamma_1',  x: \recur{\wn T}$, the assignment $x: \recur{\wn T}$ also occurs in $\Gamma_2$ which with the addition of $x:S$ in $\Gamma_2$ implies $S = \recur{\wn T}$. Notice that the case $q_1=\un$ and $q_2=\lin$ is not possible since $\un (\Gamma_1\contcomp \Gamma_2)$ implies that all assignments in $\Gamma_1\contcomp \Gamma_2$ have types $\nilT$ or with `$\un$'; thus, in that case we cannot prove $\Gamma_1 \st x: \lin \wn T.S $.

 Similarly, Rule \vastype{Out} is parametric on the qualifier $q$.

	\begin{figure}[!t]
	\begin{mdframed}
	\begin{mathpar}
		\inferrule* [ left = \vastype{Var} ]{ \contun{\Gamma} }
		{ \Gamma , x:T \st x:T }
		\and
		\inferrule* [ left = \vastype{Nil} ]{ \contun{\Gamma} }
		{ \Gamma \st \nil  }
		\and
		\inferrule*[ left =\vastype{Par} ]{ \Gamma_1 \st P \\ \Gamma_2 \st Q }
		{ \Gamma_1 \contcomp \Gamma_2 \st P \pp Q  }
		\and
		\inferrule*[ left =\vastype{Res} ]{ \Gamma , x:T ,y: \dual{T} \st P }
		{ \Gamma \st  \res {xy} P }
		\and
		\inferrule[\vastype{In} ]{ q_1(\Gamma_1 \contcomp \Gamma_2 ) \\ \Gamma_1 \st x: q_2 \wn T.S  \\ (\Gamma_2 + x: S), y:T \st P }
		{ \Gamma_1 \contcomp \Gamma_2 \st q_1\ x\inp y.P }
		\and
		\inferrule* [ left = \vastype{Out} ]{ \Gamma_1 \st x: q \oc T.S 	\\ \Gamma_2 \st v:T  \\ \Gamma_3 + x:S \st P  }
		{ \Gamma_1 \contcomp \Gamma_2 \contcomp \Gamma_3 \st \ov x\out v.P }
	\end{mathpar}
	\end{mdframed}
	\jspace
			\caption{Typing rules for \vasco (cf.~\cite{V12}).\label{f:vasorigrules}}	
			\jspace
	\end{figure}
	
The main property of the type system concerns \emph{well-formed} processes, which are defined next. 
\begin{definition}[Redexes and Well-formedness]
	A \emph{redex} is a process of the form $ \vasin{q}{x}{v}{P} \pp \vasout{y}{z}{Q}$.  
	Processes of the form $ \vasin{q}{x}{v}{P}$ and $\vasout{y}{z}{Q}$ have prefix $x$ and $y$, respectively.

	A process is \emph{well-formed} if, for each of its structurally congruent processes of the form $ \res {x_1y_1} \cdots \res {x_ny_n} (P \pp Q \pp R) $, the following conditions hold. (1)~If $P$ and $Q$ are processes prefixed at the same variable, then they are of
the same nature (input, output). (2)~If $P$ is prefixed at $x_1$ and $Q$ is prefixed at $y_1$ then $P \pp Q$ is a redex.
\end{definition}

\begin{theorem}[Properties of the Type System]
The type system satisfies the following properties (see \cite{V12} for details):
\begin{itemize}
	\item 
	If $\Gamma \st P$ and $P \equiv Q$, then $\Gamma \st Q$. 
	\item 
	If $\Gamma \st P$ and $P \vasred  Q$, then $\Gamma \st Q$.
	\item 
	If $~\st P$ then $P$ is well-formed.
\end{itemize}
%\begin{description}
%	\item[Weakening (Lemma 7.2 in~\cite{V12})]. \\
%	If $\Gamma \st P$ and $\contun{T}$ then $\Gamma, x:T \st P$.
%	\item[Strengthening (Lemma 7.3 in~\cite{V12})]. \\ 
%	Let $\Gamma \st P$ and $x \not \in \fn{P}$. If $\Gamma = \Gamma', x : T$ then $\Gamma' \st P$.
%	\item[Preservation under $\equiv$ (Lemma 7.4 in~\cite{V12})]. \\
%	If $\Gamma \st P$ and $P \equiv Q$, then $\Gamma \st Q$. 
%	\item[Preservation (Theorem 7.2 in~\cite{V12})]. \\
%	If $\Gamma \st P$ and $P \vasred  Q$, then $\Gamma \st Q$.
%	\item[Safety (Theorem 7.3 in~\cite{V12})].\\
%	If $~\st P$ then $P$ is well-formed.
%\end{description}
\end{theorem}

For technical convenience, we rely on the \emph{refined} typing rules for input and output in \Cref{f:vasrules}, which are equivalent (but more fine-grained) than those in \Cref{f:vasorigrules}.

	\begin{figure}[!t]
\begin{mdframed}
	\begin{mathpar}
		\inferrule[\vastype{Lin-In_1} ]
 		{ \Gamma_1 , x : \lin \vasreci{T}{S} \st x: \lin \vasreci{T}{S}  \\ 
		\Gamma_2 , x: S , y : T \st P}
		{ \Gamma_1 , x: \lin \vasreci{T}{S} \contcomp \Gamma_2 \st  \vasin{\lin}{x}{y}{P}   }
		\and
		\inferrule* [ left = \vastype{Lin-In_2} ]
		{  \Gamma_1 , x: \recur{\wn T} \st x: \recur{\wn T}  \\
		\Gamma_2 , x: \recur{\wn T} , y : T \st P
		}
		{  (\Gamma_1 , x: \recur{\wn T}) \contcomp (\Gamma_2 , x: \recur{\wn T}) \st  \vasin{\lin}{x}{y}{P} }
		\and
		\inferrule* [ left = \vastype{Un-In} ]
		{  \Gamma \st x: \recur{\wn T}  \\
		\Gamma , y : T \st P
		}
		{  \Gamma \st   \vasin{\un}{x}{y}{P}  }
		\and
		\inferrule* [ left = \vastype{Un-Out} ]
		{ \Gamma_1 \st x: \recur{\oc T}  \\  \Gamma_2 \st v: T  \\
		\Gamma_3 \st P
		}
		{  \Gamma_1 \contcomp	\Gamma_2 \contcomp \Gamma_3 \st   \vasout{x}{v}{P}} 
		\and
		\inferrule[\vastype{Lin-Out} ]
		{ \Gamma_1 \st x: \lin\ \vassend{T}{S}  \\  \Gamma_2 \st v: T  \\
		\Gamma_3, x:S \st P
		}
		{  \Gamma_1 \contcomp	\Gamma_2 \contcomp \Gamma_3 \st   \vasout{x}{v}{P}}  
	\end{mathpar}
	\end{mdframed}
	\jspace
		\caption{Refined typing rules for input and output.\label{f:vasrules}}	
		\jspace
	\end{figure}

We close this section by defining the class of processes \vaslang.
\begin{definition}[\vaslang]
\label{d:vaslang}
	We define $ \vaslang = \{ P \in \vasco \mid  \exists \Gamma \ s.t. \ \Gamma \st P  \}$.
\end{definition}

\begin{example}[A Non-Terminating Process in \vaslang]
\label{ex:infinite}
		Consider the process $P_{\ref{ex:infinite}} =  \vasres{xy}{( \vaspara{\vasout{y}{w}{\nil}}{\vasin{\un}{x}{z}{\vasout{y}{w}{\nil}}} )}$,   which invokes itself ad infinitum. 
	Process $P_{\ref{ex:infinite}}$ is in $\vaslang$ because $w: \nilT \st P_{\ref{ex:infinite}}$ holds with the following derivation:
		\begin{mathpar}
	\inferrule*[left=\vastype{Res}]{
				\inferrule*[left=\vastype{Par}]{ 
					\Pi \\
					\inferrule*[left=\vastype{Un-In}]{ 
						\inferrule{
							\contun{\Gamma}
						}{
							\Gamma  \st x: \recur{\wn \nilT}  
						}
						\\
						\Pi
					}
					{  \Gamma \st   \vasin{\un}{x}{z}{\vasout{y}{w}{\nil}}  }
				}
				{ \Gamma \st \vasout{y}{w}{\nil} \pp \vasin{\un}{x}{z}{\vasout{y}{w}{\nil}} }
			}
			{ 
				w: \nilT \st  \vasres{xy}{( \vaspara{\vasout{y}{w}{\nil}}{\vasin{\un}{x}{z}{\vasout{y}{w}{\nil}}} )}
			}
		\end{mathpar}		
with 	 $  \Gamma = x:\recur{\wn \nilT} , y: \recur{\oc \nilT}, w: \nilT $ and $\Pi$ is the derivation		
		\begin{mathpar}
			\inferrule*[left=\vastype{Un-Out}]{ 
				\inferrule{
					\contun{\Gamma' }
				}
				{
					\Gamma' \st y: \recur{\oc \nilT }  
				}
				\\
				\inferrule{
					\contun{\Gamma' }
				}{
					\Gamma' \st w: \nilT  
				}\\
				\inferrule{ 
					\contun{\Gamma' } 
				}
				{ 
					\Gamma' \st \nil  
				}
			}
			{  \Gamma' \st \vasout{y}{w}{\nil} } 
		\end{mathpar}
	with $ \Gamma' = x:\recur{\wn \nilT} , y: \recur{\oc \nilT}, w: \nilT, z: \nilT $.
\end{example}

\section{A Weight-based Approach to Terminating Processes}
\label{s:weight}
We move on to consider a type system that ensures termination for a class of $\pi$-calculus processes. 
Following Deng and Sangiorgi~\cite{DS06}, the type system uses \emph{weights} (or \emph{levels}) to avoid infinite reduction sequences. 
This type system will induce a class of terminating  $\vasco$ processes, denoted \lvllang (\Cref{d:lvllang}), obtained via appropriate translations on processes and types.
To ease the definition of such translations, here we define a type system that mildly modifies the system of~\cite{DS06} to account for linearity and synchronous/polyadic (tuple-based) communication. 
Our main result is that the weight-based system ensures termination (\Cref{t:lvlsn}). %, closely following Deng and Sangiorgi's proof. 

\subsection{Processes} 
We introduce a process model for the weight-based type system, denoted \pilvl, formally defined next.
In the following, we write $\tilde{y}$ to stand for the finite tuple $y_1,\cdots, y_n$. 

\begin{definition}[Processes]
The syntax of \pilvl processes is given by the grammar in \Cref{f:lvlsyntandtype} (top).
\end{definition}

\pilvl is designed to stand in between \vasco and the process model in~\cite{DS06}.  
Communication in \pilvl is polyadic, i.e., exchanges involve a tuple of names, rather than a single name as in  \Cref{def:vascosyntax} and~\cite{DS06}. 
We shall often consider tuples of length two (i.e., dyadic communication), as this suffices for a continuation-passing encoding of sessions~\cite{DGS12}. 
Another difference with respect to~\cite{DS06} is that inputs can be linear or  unrestricted; this will facilitate the formal connection with \vasco and its type system. 
The role of linearity is more prominent at the level of types, defined later on.

We give the operational semantics of \pilvl in terms of the (early)
labeled transition system (LTS), with the following labels for input, output, bound output, and silent transitions (synchronizations):
$$\alpha ::= x(\tilde{v}) \,|\, \ov x\out { \tilde{y} } \,|\, \res {y,\widetilde{b}}\ov x\out { \tilde{v} } \,|\, \tau $$ 
%\daniele{DN: It seems that there are some labels missing, or something should be said about either $\tilde{b}$ or $\tilde{y}$ being `empty' (cf. \Cref{f:lvllts}). }

The rules, given in \Cref{f:lvllts}, are standard. Rules \lvltype{Par} and \lvltype{Tau} can be applied symmetrically across parallel composition.
	
\begin{figure*}[!h]
	\begin{mdframed}
	\begin{mathpar}
		\inferrule[\lvltype{In} ]{  }
		{ \lvlinpoly{x}{y}{z}{P} \lts{x \inp {\tilde{v}}} P\substj{v_1}{y_1}\substj{v_2}{y_2}  }
		\and
		\inferrule[\lvltype{Par} ]{ P \lts{\alpha} P' \\ \bn{\alpha} \cap \fn{Q} = \emptyset }
		{ P \pp Q \lts{\alpha} P'\pp Q }
		\and
		\inferrule[\lvltype{Res} ]{ P \lts{ \alpha } P' \\ x \not \in \names{\alpha} }
		{ \res {x}P \lts{ \alpha } \res {x}P' }
				\and
		\inferrule[\lvltype{Rep} ]{  }
		{ \bang x \inp { \tilde{y} }.P \lts{x \inp { \tilde{v} }} \bang x \inp { \tilde{y}}.P  \pp P\substj{v_1}{y_1} }

		\\
				\inferrule[\lvltype{Out} ]{   }
		{ \lvloutpoly{x}{y}{z}{P}  \lts{\ov x\out { \tilde{y} }}  P  }
		\and
		\inferrule[\lvltype{Tau} ]{ P  \lts{ \res {\widetilde{b}}\ov x\out { \tilde{v} } } P' \\ Q \lts{ x \inp {\tilde{v}} }  Q' \\ \widetilde{b} \cap \fn{Q} = \emptyset }
		{ P \pp Q \lts{ \tau } \res{\widetilde{b}} (P'\pp Q')}
		\and
		\inferrule[\lvltype{Open} ]{ P \lts{ \res {\widetilde{b}}\ov x\out { \tilde{v} } } P' \\ y \in (\fn{v_1} \cup \fn{v_2}) - \{ \widetilde{b} , x \} }
		{ \res {y}P \lts{ \res {y,\widetilde{b}}\ov x\out { \tilde{v} } } P' }
	\end{mathpar}
	\end{mdframed}
	\jspace
\caption{An LTS for \pilvl}
\label{f:lvllts}
\jspace
\end{figure*}

\subsection{Types}

\begin{definition}[Types for \pilvl]
The syntax of weight-based types for \pilvl is given by the grammar in \Cref{f:lvlsyntandtype} (bottom).
\end{definition}

As in~\cite{DS06}, our link types  for \pilvl are \emph{simple}, i.e., they do not admit the sequencing of actions enabled by session types. 
Our syntax of types extends that in~\cite{DS06} to account for (i)~dyadic communication and (ii)~explicit types for clients and servers.
Concerning~(ii), we purposefully adopt the tail-recursive types for clients and servers defined for \vasco, rather than more general recursive types.

\begin{figure}[t!]
	\begin{mdframed}
	\[
	\begin{aligned}
		P,Q ::= &								&& \mbox{(Processes)}\\
				&  \lvlout{x}{y_1}{y_2}{P}		&& \mbox{(output)}
				%&  \ov x^{\un} \out { y^q,z^p }.P	&& \mbox{(unrestricted output)}\\
				&&  \lvlin{x}{y_1}{y_2}{P}					&& \mbox{(linear input)}
				\\
								& P \pp Q  	    		  		&& \mbox{(parallel)}
				&&  \lvlserv{x}{y_1}{y_2}{P}			&& \mbox{(server)}\\
				%&  \ov x^{p} \out { y^q,z^p }.P		&& \mbox{(output)}\\
				%&  \ov x^{\un} \out { y^q,z^p }.P	&& \mbox{(unrestricted output)}\\
				%&  x^p\inp {y^q , z^p}.P				&& \mbox{(linear input)}\\
				%&  \bang x^p \inp {y^q , z^p }.P		&& \mbox{(unrestricted input)}\\
				& \res {x}P						&& \mbox{(restriction)}
				&& \nil	     					&& \mbox{(inaction)}\\
	\end{aligned}
	\]
	\hrule
	\smallskip
	\[
	\begin{aligned}
		S,T,V,L ::= &								&& \mbox{(Link Types)}\\
		%		& V								&& \mbox{(value)} \\
		%		& L								&& \mbox{(link)} \\
		%V	::= &								&& \mbox{(value types)}\\
		%		& L								&& \mbox{(link)}\\
				%& 								&& \mbox{(reserved for any other types such as bool, nat)}\\
		%L   ::= &								&& \mbox{(link types)}\\
				& \lvlint{n}{V_1}{V_2}			&& \mbox{(linear input type)}\\
				& \lvloutt{n}{V_1}{V_2}			&& \mbox{(linear output type)}\\
				& \lvlunst{n}{V}				&& \mbox{(unrestricted server type)}\\
				& \lvlunct{n}{V}				&& \mbox{(unrestricted client type)}\\
				%& q \#^n \widetilde{V} = 
				%	\begin{cases}
				%		q = \ast \quad \widetilde{V} = V \\
				%		q = \circ \quad \widetilde{V} = (V_1 , V_2)
				%	\end{cases}								&&
				%	\begin{aligned}
				%	& \mbox{(unrestricted link type)} \\ 
				%	& \mbox{(linear link type)}
				%	\end{aligned}
				%\\
				%& \level{n}{V_1}{V_2} 			&& \mbox{(linear comunication)}\\
				%& \rlevel{n}{V}	 				&& \mbox{(unrestricted comunication)}\\
				& \lvlnilT							&& \mbox{(termination)}
		%v	::= &								&& \mbox{(Values)}\\
		%		& x								&& \mbox{(variables)}\\
		\\
		n ::= & ~1, 2, \cdots && \mbox{(weights)}
	\end{aligned}
	\]
	\end{mdframed}
	\jspace
	\caption{Syntax of processes and types for \pilvl.}
	\label{f:lvlsyntandtype}
	\jspace
\end{figure}

We introduce some notions borrowed from the  type system from \Cref{ss:typesess}: duality, contexts, predicates on types, operations on contexts.

\begin{definition}[Duality]
	Duality on linked types is defined as:
		\[
		\begin{aligned}
			\dual{\lvlint{n}{V_1}{V_2}}	& = \lvloutt{n}{\dual{V_1}}{\dual{V_2}}		&
			\dual{\lvloutt{n}{V_1}{V_2}} & = \lvlint{n}{\dual{V_1}}{\dual{V_2}}	& \dual{\lvlnilT} & = \lvlnilT	
			\\
			\dual{\lvlunst{n}{V}} & = \lvlunct{n}{\dual{V}}				
			&
			\dual{\lvlunct{n}{V}} & = \lvlunst{n}{\dual{V}}				
			%\\
			%\dual{\lvlnilT} & = \lvlnilT\\
			%\tilde{\dual{V}} & = \dual{V_1} , \dual{V_2} \quad \text{where } \tilde{V} = V_1 , V_2\\
		\end{aligned}	
	\]
%Sometimes we use the abbreviations $\tilde{V} = V_1 , V_2$ and  $\tilde{\dual{V}}  = \dual{V_1} , \dual{V_2}$.  %\daniele{I would prefer not to have abbreviations in the type system or in the language. They don't save much space and having them explicit helps to read the rules.}
\end{definition}

\begin{definition}[Contexts]
	Contexts are given by the grammar:
	\[
		\Gamma , \Delta :: = \cdot \sep \Gamma , x:V \sep \Gamma , x:\dualjoin{V}{\dual{V}}
	\]
	where $\Gamma , x:L$ and $\Gamma , x:\dualjoin{L}{\dual{L}}$ imply $x \not \in \dom{\Gamma}$. 
	
	Following the sorts of~\cite{HVK98}, the assignment 
	$x:\dualjoin{L}{\dual{L}}$ denotes the pairing of $x$  with two complementary protocols, where $\dualjoin{L}{\dual{L}} = \dualjoin{\dual{L}}{L}$. 
	We use \lvlsplit{x}{L} to stand for $x:\dualjoin{L}{\dual{L}}$ when $L$ is the main object of interest.
	We write $x\lvlass T $ if either $x:T$ or $x::T$ holds (i.e., $\lvlass \in \{ : , :: \}$). 
	%\daniele{DN:I think the notation with $::$ is confusing. I would prefer to have the type written explicitly. }
\end{definition}

\begin{definition}[Unrestricted Types]
	Predicate $\contun{T}$ holds if $T = \lvlunst{n}{V}$, $T = \lvlunct{n}{V}$, $T = \lvlnilT$, or x:\dualjoin{L}{\dual{L}} with $\contun{L}$. We write $\contun{\Gamma}$ if $\contun{T}$ holds for every $x \lvlass T \in \Gamma$.
\end{definition}

%q \#^n \widetilde{V} = 
%				\begin{cases}
%					q = \ast \quad \widetilde{V} = V \\
%					q = \circ \quad \widetilde{V} = (V_1 , V_2)
%\circ \#^n (V_1 , V_2)
%\ast \#^n V 

Following \Cref{d:splitsts}, the following definitions 
gives a relation to split contexts into two parts.

\begin{definition}[Split Relation on Contexts]
The relation $\contcomp$ on contexts is defined in \Cref{f:lvlcontrelation}. \end{definition}

\begin{figure}
	\begin{mdframed}
	\begin{mathpar}
		\inferrule{}
		{ \emptycont \contcomp \emptycont = \emptycont }
		\and
		\inferrule{ \Gamma_1 \contcomp \Gamma_2 = \Gamma  \\ \contun{T}}
		{\Gamma , x: T = ( \Gamma_1 , x: T  ) \contcomp ( \Gamma_2 , x: T  ) }
		\and
		\inferrule{ \Gamma_1 \contcomp \Gamma_2 = \Gamma  \\ \contun{V}}
		{\Gamma , x:\dualjoin{V}{\dual{V}} = ( \Gamma_1 , x:\dualjoin{V}{\dual{V}}  ) \contcomp ( \Gamma_2 , x:\dualjoin{V}{\dual{V}}  ) }
		\and
		\inferrule{ \Gamma_1 \contcomp \Gamma_2 = \Gamma \\ \neg \contun{V}}
		{ \Gamma , x:\dualjoin{V}{\dual{V}} = (\Gamma_1 , x: V) \contcomp  (\Gamma_2 , x: \dual{V}) }
		\and
		\inferrule{ \Gamma_1 \contcomp \Gamma_2 = \Gamma \\ \neg \contun{V}}
		{ \Gamma , x:\dualjoin{V}{\dual{V}} = (\Gamma_1 , x:\dualjoin{V}{\dual{V}}) \contcomp  \Gamma_2  }
		\and
		\inferrule{ \Gamma_1 \contcomp \Gamma_2 = \Gamma \\ \neg \contun{V}}
		{ \Gamma , x:\dualjoin{V}{\dual{V}} = \Gamma_1  \contcomp  (\Gamma_2 , x:\dualjoin{V}{\dual{V}}) }
		\and
		\inferrule{ \Gamma_1 \contcomp \Gamma_2 = \Gamma \\ \neg \contun{T}}
		{ \Gamma , x: T = (\Gamma_1 , x: T) \contcomp \Gamma_2 }
		\and
		\inferrule{  \Gamma_1 \contcomp \Gamma_2 = \Gamma \\ \neg \contun{T}}
		{ \Gamma , x: T = \Gamma_1 \contcomp ( \Gamma_2 , x: T ) }
	\end{mathpar}
	\end{mdframed}
	\jspace
\caption{Splitting of Contexts for \pilvl}
\label{f:lvlcontrelation}
\jspace
\end{figure}

We now introduce notions on processes that are essential to Deng and Sangiorgi's approach to termination by typing.

\begin{definition}[Level Function, $\levelof{x}$]
Let $\mathcal{N}$ denote the set of all names. 
We define the function $\levelof{\cdot}: \mathcal{N} \to \mathbb{N}$ to map names of a process (free and bound) to naturals. 
We assume $\alpha$-conversion is silently used to avoid name capture and ensure uniqueness of bound names.
Given a (typed) process, we define this function as follows: 
%\daniele{DN: I think we should write this definition explicitly, without the $\diamond$; it is an abbreviation of an abbreviation.  }
$$
\levelof{x} = \begin{cases}
%	n & \text{if $x \diamond \lvlint{n}{V_1}{V_2},x\diamond\lvloutt{n}{V_1}{V_2} , x\diamond\lvlunst{n}{V} $ or $x\diamond\lvlunct{n}{V} $}
	n & \text{if $x : T$ or $x ::T$}
	\\
	& \text{with $T \in \{\lvlint{n}{V_1}{V_2},\lvloutt{n}{V_1}{V_2}, \lvlunst{n}{V}, \lvlunct{n}{V}\}$}
	\\
	m & \text{if $x: \lvlnilT$, for any $m \in \mathbb{N}$} 
\end{cases}
$$
\end{definition}

\begin{definition}[Active Outputs, $\outsubj{\cdot}$]
Given a process $P$, the set of names with active outputs $\outsubj{P}$ is  defined inductively:
\[
\begin{aligned}
	\outsubj{\lvloutpoly{x}{y}{z}{P}} & = \{x\} \cup \outsubj{P}
	&
	\outsubj{\lvlinpoly{x}{y}{z}{P}} & = \outsubj{P}\\
	\outsubj{P \pp Q } & = \outsubj{P} \cup \outsubj{Q}
	&
	\outsubj{\res {x}P	} & = \outsubj{P}\\
		\outsubj{\nil} & = \emptyset 
	&
	\outsubj{\lvlservpoly{x}{y}{z}{P}} & = \emptyset
\end{aligned}
\]
\end{definition}

%\begin{notation}
%	We use the notation $\tilde{y} \lvlass \tilde{V}$ to represent $ y_1 \lvlass V_1 , y_2 \lvlass V_2$ when $\tilde{y} = y_1, y_2$ and $ \tilde{V} = V_1,V_2$. \daniele{DN: same remark as above.}
%\end{notation}

\begin{figure*}[!t]
	\begin{mdframed}
		\begin{mathpar}
			\inferrule[\lvltype{Var_1} ]{ \contun{\Gamma} }
			{ \Gamma , x:V \lt x:V }
			\and
			\inferrule[\lvltype{Var_2} ]{ \contun{\Gamma} }
			{ \Gamma , x: \dualjoin{V}{\dual{V}} \lt x:V }
			\and
						\inferrule[\lvltype{Nil} ]{ \contun{\Gamma} }% \contun{\Gamma} }
			{ \Gamma \lt \nil }
			\and
			\inferrule[\lvltype{Par} ]{\Gamma_1 \lt P \\ \Gamma_2 \lt Q }
			{ \Gamma_1 \contcomp \Gamma_2 \lt P \pp Q }
			\and
			\inferrule[\lvltype{Res} ]{ \Gamma, x: \dualjoin{V}{\dual{V}}  \lt P  }
			{ \Gamma \lt \res {x}P }
			\and 
			\inferrule* [ left = \lvltype{Lin-In_1} ]{ \Gamma_1 , x: \lvlint{n}{V_1}{V_2} \lt x: \lvlint{n}{V_1}{V_2} \\ \Gamma_2 ,  y_1:V_1, y_2:V_2 \lt P \\ \levelof{x} = \levelof{y_2} }
			{ \Gamma_1 , x: \lvlint{n}{V_1}{V_2} \contcomp \Gamma_2  \lt \lvlin{x}{y_1}{y_2}{P} }
			\and
			\inferrule* [ left = \lvltype{Lin-In_2} ]{ 
				\Gamma_1 , x: \lvlunst{n}{V} \lt x: \lvlunst{n}{V} \\
				\Gamma_2 , x: \lvlunst{n}{V} , y_1:V , y_2: \lvlnilT \lt P }
			{ ( \Gamma_1 , x: \lvlunst{n}{V}) \contcomp (\Gamma_2 , x: \lvlunst{n}{V}) \lt \lvlin{x}{y_1}{y_2}{P}  }
			\and
			\inferrule* [ left = \lvltype{Lin-In_3} ]{
				\Gamma , \lvlsplit{x}{\lvlunst{n}{V}} \lt x:\lvlunst{n}{V} \\
				\Gamma , \lvlsplit{x}{ \lvlunst{n}{V}} ,   y_1:V  , y_2: \lvlnilT \lt P }
			{ \Gamma , \lvlsplit{x}{\lvlunst{n}{V}} \lt \lvlin{x}{y_1}{y_2}{P}  }
			\and
			\mprset{ flushleft }
			\inferrule* [ left = \lvltype{Lin-Out} ]{
				\Gamma_1 ,  x: \lvloutt{n}{V_1}{V_2} \lt  x: \lvloutt{n}{V_1}{V_2} \\
				\Gamma_2 , y_1:V_1 \lt y_1:V_1 \\
				\Gamma_3 , y_2:V_2 \lt P \\ \levelof{x} = \levelof{y_2} }
			{ (\Gamma_1 , x: \lvloutt{n}{V_1}{V_2}) \contcomp (\Gamma_2 , y_1:V_1) \contcomp (\Gamma_3, \lvlsplit{y_2}{ V_2} ) \lt \lvlout{x}{y_1}{y_2}{P} }
			\and
			\inferrule* [ left = \lvltype{Un-Out_1} ]{ 
				\Gamma_1 , x:  \lvlunct{n}{V} \lt x:  \lvlunct{n}{V} \\
				\Gamma_2, x:  \lvlunct{n}{V} , y: V \lt y: V \\
				\Gamma_3, x:  \lvlunct{n}{V} , y_2: \lvlnilT   \lt P  
			}
			{( \Gamma_1 , x:  \lvlunct{n}{V} )\contcomp	(\Gamma_2, x:  \lvlunct{n}{V} , y: V) \contcomp (\Gamma_3, x:  \lvlunct{n}{V})
			 \lt 
			\lvlout{ x}{y_1}{y_2}{P} 
			}
			\and
			\inferrule* [ left = \lvltype{Un-Out_2} ]{  
				\Gamma_1, \lvlsplit{x}{ \lvlunct{n}{V}} \lt x:  \lvlunct{n}{V} \\
				\Gamma_2, \lvlsplit{x}{ \lvlunct{n}{V}} , y: V \lt y: V \\
				\Gamma_3 , \lvlsplit{x}{ \lvlunct{n}{V}} , y_2: \lvlnilT   \lt P  
			}
			{ (\Gamma_1 , \lvlsplit{x}{ \lvlunct{n}{V}}) \contcomp	
			(\Gamma_2, \lvlsplit{x}{ \lvlunct{n}{V}} , y: V) \contcomp 
			(\Gamma_3, \lvlsplit{x}{ \lvlunct{n}{V}}) \lt 
			\lvlout{x}{y_1}{y_2}{P} 
			}
			\\
			\inferrule* [ left = \lvltype{Un-In_1} ]{ 
				\Gamma, x: \lvlunst{n}{V} \lt x: \lvlunst{n}{V} \\
				\Gamma, x: \lvlunst{n}{V}, y_1:V , y_2: \lvlnilT \lt P \\  \forall b \in \outsubj{P}, \ \levelof{b} < n }
			{ \Gamma, x: \lvlunst{n}{V} \lt  \lvlserv{x}{y_1}{y_2}{P} }
			\and
			\inferrule* [ left = \lvltype{Un-In_2} ]{
				\Gamma, \lvlsplit{x} {\lvlunst{n}{V}} \lt x: \lvlunst{n}{V} \\
				\Gamma, \lvlsplit{x} {\lvlunst{n}{V}}, y_1:V , y_2: \lvlnilT \lt P \\  \forall b \in \outsubj{P}, \ \levelof{b} < n }
			{ \Gamma, \lvlsplit{x}{\lvlunst{n}{V}} \lt  \lvlserv{x}{y_1}{y_2}{P} }
		\end{mathpar}
	\end{mdframed}
	\jspace
\caption{Typing rules for \pilvl}
\label{f:lvltyperules}
\jspace
\end{figure*}

Typing judgments are of the form $\Gamma \lt P$, with corresponding typing given in \Cref{f:lvltyperules}.  
%\daniele{DN: we need to say somewhere that ``, '' binds strongly than $\contcomp$, otherwise the rules that use $\contcomp$ will be ambiguous.}
Typability is contingent on a level function: we say a process $P$ is well-typed if there exists a level function $\levelof{\cdot}$ such that a typing derivation $\Gamma \lt P$ holds, for some $\Gamma$.

%\jpedit{Need explanations for the typing rules!}

%\joe{Explonation of typing rules}

%\lvltype{Var_1} mimics the rule \vastype{Var} however \lvltype{Var_2} is the equivalent rule for complimentary interaction, here if $x:\dualjoin{V}{\dual{V}}$ then we may type $x:V$. The intuition here is channel $x$ contains the type of its two endpoints, these being $v$ and $\dual{V}$, as long as $x$ expects on of these types then the channel is typed correctly.
%$\lvltype{Lin-In_1}$ is the linear counterpart of $\vastype{In}$, notice that notice there is no counterpart for $x$ being a linear complimentary interaction, instead the context split $\Gamma , x:\dualjoin{V}{\dual{V}} = (\Gamma_1 , x:{V}) \contcomp  (\Gamma_2 , x:{\dual{V}})$ allows us to apply $\lvltype{Lin-In_1}$. This structure acts silently within the rules where $V$ is linear via the splitting of contexts, disallowing linear channels to consume linear complimentary interactions. The rules $\lvltype{Lin-In_2}$ , $\lvltype{Lin-In_3}$ are the counterparts to $\vastype{In}$ when the type is unrestricted. Similarly, $\lvltype{Lin-Out}$ , $\lvltype{Un-Out_1}$ , $\lvltype{Un-Out_2}$ represent the rule $\vastype{Out}$. $\lvltype{Un-In_1}$ and $\lvltype{Un-In_2}$ are the unrestricted counterpart to $\vastype{In}$, here there is a restriction that the levels of the active outputs must be strictly less than the channel of the server providing them. Finally, $\lvltype{Res}$ illustrates how a restricted channel is typed via a complimentary interaction.

We comment on some of the rules in \Cref{f:lvltyperules} for \pilvl, contrasting them with those in \Cref{f:vasorigrules} for \vasco.
Rule \lvltype{Var_1} is similar to rule~\vastype{Var}.
Rule~\lvltype{Var_2} is the corresponding rule for complementary interaction:   if $x:\dualjoin{V}{\dual{V}}$, then we can assign the type $x:V$. Intuitively, name $x$ encapsulates the types of its two endpoints, denoted as $V$ and $\dual{V}$. As long as $x$ respects one of these types, the channel is considered correctly typed.

Rule \lvltype{Lin-In_1} acts as the linear counterpart to \vastype{In}. Importantly, there is no direct counterpart for $x$ as a linear complementary interaction. Instead, the context split $\Gamma , x:\dualjoin{V}{\dual{V}} = (\Gamma_1 , x:{V}) \contcomp (\Gamma_2 , x:{\dual{V}})$ allows for the application of rule~\lvltype{Lin-In_1}. This structural mechanism operates silently within the rules where $V$ is linear, achieved through context split. As a result, this disallows linear channels from consuming linear complementary interactions.

Rules \lvltype{Lin-In_2} and \lvltype{Lin-In_3}, the first with `:' and the second with `::',  are  counterparts to rule \vastype{In} for  unrestricted types with linear qualifier. 
Similarly, \lvltype{Lin-Out}, \lvltype{Un-Out_1}, and \lvltype{Un-Out_2} represent the rule \vastype{Out}. 
Furthermore, \lvltype{Un-In_1} and \lvltype{Un-In_2} are the unrestricted counterparts to rule \vastype{In} with unrestricted qualifier. 
These rules adopt the main condition from~\cite{DS06}, i.e., the weight of types of the active outputs must be strictly less than the weight of the type of the channel of the server providing them. 
Finally, rule \lvltype{Res} types a restricted channel through a complementary interaction.

We state the type preservation property:
\begin{theorem}[Type Preservation]
	\label{thm:lvltypepres}
	Suppose $\Gamma \lt P$ for a level function $l$. If $ P \lts{\tau} P' $ then $\Gamma \lt P'$ for the same level function $l$.
\end{theorem}

%\begin{proof}
%	By induction on the structure of $P$.
%\end{proof}

\subsection{Termination by Typing}
A process terminates if all its reduction sequences are finite. 
We show that our formulation of the type system in \cite{DS06} also enforces termination by typing. The proof follows the same lines as in~\cite{DS06}: 
a weight is associated with a well-typed process; this weight is then shown to strictly reduce when the the process synchronizes. 
The weight is actually a \emph{vector} constructed from the observable active outputs  of a channel within a typed process. 
%We formally define .

\begin{definition}[Vectors]
	We define vectors and their operations: % used for measuring the weight of a process:
\begin{itemize}
	\item Given $k \geq 1$, we write $\zerovec_i$ to denote the vector $\langle n_k , n_{k-1} , \cdots ,  n_1  \rangle $ where $n_i = 1$ and $ n_j = 0$ for every other~$j$. Also, $\zerovec$ denotes the zero vector where  $n_i = 0$ for every $i$. 
	\item 
	Given vectors $v_1 =\langle n_k , n_{k-1} ,\cdots ,  n_1  \rangle$ and $v_2 = \langle m_l , m_{l-1}  ,\cdots ,  m_1  \rangle $,  with $k \geq l$, the sum  $v_1+v_2$ is defined in two steps. 
	Firstly, if $k > l$ then the shorter vector $v_2$ is extended into $v'_2$ by adding  zeroes to match the size of $v_1$, i.e., $v'_2 =\langle m_k , m_{k-1} ,\cdots , m_l  ,\cdots ,  m_1  \rangle $, with $ \langle m_k , m_{k-1} ,\cdots , m_{l+1} \rangle = \zerovec$. 
	Then, addition of $ v_1 $ and $ v'_2 $ is applied pointwise.
	
	\item Given vectors $v_1 = \langle n_k , n_{k-1} ,\cdots ,  n_1 \rangle$ and $v_2 = \langle m_k , m_{k-1} ,\cdots ,  m_1 \rangle$ of equal size $k$, the ordering $v_1 \wtless v_2$ is defined iff $\exists i \leq k, \ n_i < m_i $ and $\forall j > i, \ n_j = m_j$.
	\end{itemize}
\end{definition}

Using vectors, we define the weight of a well-typed process:

\begin{definition}[Weights]
	Given a well-typed process $P$ with level function $l$, the weight of  $P$ is the vector defined inductively as: 
	\[
	\begin{aligned}
	\weight{\nil} & = 0 &
	\weight{\bang x \inp { \tilde{y}}.P	}&  = 0 
	\\
	\weight{x\inp { \tilde{y}}.P} & = \weight{P} &
	\weight{\ov x\out { \tilde{y} }.P} & = \weight{P} + \zerovec_{\levelof{x}} 
	\\
	\weight{P \pp Q} & = \weight{P} + \weight{Q} &
	\weight{\res {x}P} & = \weight{P} &
	\end{aligned}
	\]	
\end{definition}		

We have the following  results, whose proof is as in~\cite{DS06}:

\begin{proposition}%[Weights Strictly Reduce in Computation]
	\label{prop:lvlweightreduce}
	If $ \Gamma \lt P $ and $ P \lts{\tau} P' $ then $ \weight{P'} \wtless \weight{P}$.
\end{proposition}

%We may now state the main result:

\begin{theorem}[Termination]
\label{t:lvlsn}
	If $\Gamma \lt P$ then $P$ terminates.
\end{theorem}

%\begin{proof}
%	Let us consider the following two cases:
%	\begin{enumerate}
%		\item When $ \weight{ P } = \zerovec $ then $P$ has no active outputs. Therefore, $P$ has no enabled synchronizations, and so no reductions can take place. Thus, $P$ terminates.
%		\item When $ \zerovec \wtless \weight{ P } $ then the set of active inputs is non-empty. If none of the channels in the set of active inputs enables a synchronization then the process must terminate. If syncronizations are enabled then, by \Cref{prop:lvlweightreduce}, any synchronization strictly reduces the weight of the reduced process. As processes are of finite length and hence weights are finite we may inductively reduce until we reach a process of weight $\zerovec$ or such that the set of active outputs do not allow for synchronizations.
%	\end{enumerate}
%\end{proof}

\section{\lvllang: A Class of Terminating Processes}
\label{s:lvllang}

Here we define and study \lvllang, a class of terminating \vasco processes induced by the weight-based type system given in \Cref{s:weight}, which leverages  
translations on processes and types/contexts, denoted $\sttoltp{\cdot}{}$ and $\sttoltt{\cdot}{l}{}$, respectively.
Concretely, \lvllang is defined as follows:

\begin{definition}[$\lvllang$]
\label{d:lvllang}
We define:
 $$ \lvllang = \{ P \in \vasco \mid  \exists \Gamma , l \ s.t. \ (\Gamma \st P) \land  \ %\sttoltj{\Gamma \st P}{l}{} = 
\sttoltt{\Gamma}{l}{} \lt \sttoltp{P}{}     \} $$
\end{definition}	

Hence, \lvllang contains those processes from $\vaslang$ (\Cref{d:vaslang}) whose translation gives typable $\pilvl$ processes.
By \Cref{t:lvlsn},   $\lvllang$ thus provides a characterization of terminating processes in $\vaslang$. 
In the following we formally define the translations $\sttoltp{\cdot}{}$ and $\sttoltt{\cdot}{l}{}$, and establish their main properties. 
Our main result is that $\lvllang \subset \vaslang$ (\Cref{t:wsubsets}): there are typable processes in $\vaslang$ which are not terminating under the weight-based approach.

\subsection{The Typed Translation}
Our translation is \emph{typed}, i.e., the translation of a \vasco process depends on its associated (session) types. 
We first present the translation on processes and types separately; then, we   combine them to define the translation of a typing judgment.
\begin{definition}[Translating Processes]
\label{d:encvs}
The translation   $\sttoltp{\cdot}{}: \vasco \to \pilvl $ is given in \Cref{f:lvltypeencod} (top), where we assume $z$ is fresh.
\end{definition}
 
We discuss some interesting cases in the translation of processes:
\begin{itemize} 
\item The shape of process $\sttoltp{  \vasin{\lin}{x}{y}{P}  }{} $ depends on whether  $x$ has a linear or an unrestricted type: this is due to rule $\vastype{In}$ (\Cref{f:vasorigrules}) which depends on a qualifier $q_2$ that can be linear or unrestricted. If $x: \lin \wn T.S$ then the translation is  $\lvlin{x}{y}{z}{\sttoltp{P\substj{z}{x}}{}}$, with the continuation along $z$; otherwise, in case $x: \recur{\wn T}$, the translation is $\lvlin{x}{y}{z}{\sttoltp{P}{}}$, since there is no continuation in $x$, as explained in the description of rule \vastype{In} in \Cref{f:vasorigrules}.
% In fact, by inspecting rule $\vastype{In}$ in \Cref{f:vasorigrules}, for the case $q_2=\un$, we would have $x: \lin \wn T.S$ in $\Gamma_2$ and have to add $x:S$ in $\Gamma_2$, as in $(\Gamma_2+ x:S)$, to type $P$, which is not possible. Therefore, $P$ would have no continuation on $x$ and the type of $x$ would be irrelevant in typing $P$.
 \item The process $\sttoltp{   \vasin{\un}{x}{y}{P}  }{} $ is simply an unrestricted input process $\lvlserv{x}{y}{z}{\sttoltp{P}{}} $.
\item The process $\sttoltp{ \vasout{x}{y}{P}  }{} $, the translation of a bound send, also depends on the type of $x$ and the justification for it is similar to the translation of linear inputs described above. 
\item The process $ \sttoltp{\res {xy}P}{}  $ is simply $ \res {z} \sttoltp{P\substj{z}{x}\substj{z}{y}}{}$: the co-variables $x,y$ are replaced by the restricted (fresh) name $z$. The duality between the types of $x$ and $y$, say $x:L$ and $y:\dual{L}$, must be preserved by the type of $z$ in $\pilvl$. This correspondence will become evident when discussing the translation of judgements (\Cref{def:firsttransl_judgments}).

\end{itemize}

\begin{figure}[!t]
	\begin{mdframed}
	\[
	\begin{aligned}
		\sttoltp{ \nil  }{} &=  \nil  \\
		\sttoltp{  P \pp Q  }{} 	&=  \sttoltp{P}{} \pp \sttoltp{Q}{}  \\
		\sttoltp{  \res {xy} P  }{} 	&=   \res {z} \sttoltp{P\substj{z}{x}\substj{z}{y}}{}  \\
		\sttoltp{  \vasin{\lin}{x}{y}{P}  }{} 	&=  
				\begin{cases}
						\lvlin{x}{y}{z}{\sttoltp{P\substj{z}{x}}{}} & \text{If } x: \lin \wn T.S \\
						\lvlin{x}{y}{z}{\sttoltp{P}{}} & \text{If } x: \recur{\wn T}\\
				\end{cases}  \\
				\sttoltp{   \vasin{\un}{x}{y}{P}  }{} 	&=  \lvlserv{x}{y}{z}{\sttoltp{P}{}}  \\
				\sttoltp{ \vasout{x}{y}{P}  }{} 	&= 
				\begin{cases}
						\res{z} \lvlout{x}{y}{z}{\sttoltp{P\substj{z}{x}}{}} & \text{If }x: \lin \oc T.S \\
						\lvlout{x}{y}{z}{\sttoltp{P}{}} & \text{If }x: \recur{\oc T} \\
				\end{cases}
	\end{aligned}
	\]
		\hrule
		\smallskip
		\[
	\begin{aligned}
		\sttoltt{ \Gamma }{l}{} 	&= \sttoltt{ x_1 : T_1}{l}{} , \cdots , \sttoltt{x_n : T_n }{l}{}\\
		\sttoltt{x : T}{l}{}		&= x : \sttoltt{ T }{l}{x}\\
		\sttoltt{\nilT}{l}{x} &=  \lvlnilT \\
		\sttoltt{ \lin\  \vasreci{T}{S} }{l}{x}	&= \lvlint{\levelof{x}}{\sttoltt{T}{l}{\alpha}}{\sttoltt{S}{l}{x}}\\
		\sttoltt{ \lin\  \vassend{T}{S} }{l}{x}	&= \lvloutt{\levelof{x}}{\sttoltt{T}{l}{\beta}}{\sttoltt{S}{l}{x}}\\
		\sttoltt{ \recur{\wn T} }{x}{l}	&= \lvlunst{\levelof{x}}{\sttoltt{T}{l}{\gamma}}\\
		\sttoltt{ \recur{\oc T} }{x}{l}	&=\lvlunct{\levelof{x}}{\sttoltt{T}{l}{\gamma}}\\
	\end{aligned}
	\]
	\end{mdframed}
	\jspace
\caption{From \vasco to \pilvl  (\Cref{d:encvs,d:encvst})}
\label{f:lvltypeencod}
\jspace
\end{figure}

\begin{definition}[Translating Types/Contexts]
\label{d:encvst}
	The translation  $\sttoltt{-}{l}{}$ of session types and contexts is given in \Cref{f:lvltypeencod} (bottom). 
	The translation of  contexts
is parametric on a level function~$l$.  In particular, the translation of a type assignment $\sttoltt{x:T}{l}{}$, relies on an auxiliary translation $x:\sttoltt{T}{l}{x}$, which is deemed to be assigned a level $l(x)$ in the translated type $\sttoltt{T}{l}{x}$, depending on the shape of $T$. Other names, denoted $\alpha, \beta,\gamma\ldots$, are necessary when translating within types.
\end{definition}

The translation $\sttoltt{-}{l}{}$ follows  the continuation-passing approach of~\cite{DGS12} to encode session types into link types. 
The translation of tail-recursive types is rather direct, and self-explanatory.

%As expected, linear receive/send (tail-recursive) session types are translated to linear input/output types (respectively) in $\pidill$ in which a level is assigned to the name carrying the input/output behaviour. Similarly, unrestricted receive/send (tail-recursive) session types are translated to 
%unrestricted server/client types in $\pidill$, respectively.

By combining the translations of types and processes in \Cref{f:lvltypeencod} we obtain a translation of type judgements / derivations in $\vasco$ into type judgements / derivations in $\pilvl$. 
We use an auxiliary notation:
%
%\begin{notation}
%	We write 
%	$ \lvljoinsub{P}{z}{x,y}$ to stand for $P\substj{z}{x}\substj{z}{y}$.
%\end{notation}

%\begin{theorem}[Substitution Lemma]\label{prop:sublvl}
%Suppose $\Gamma , u: T , v : S \lt P$, with $T = \dual{S}$	and $ \contun{T}$ for a level function $l$. Then  $\Gamma , z: \dualjoin{T}{S} \lt \lvljoinsub{P}{z}{u,v}$ for $l$.
%\end{theorem}

%

%\daniele{DN: We need to move the definition of $\lvllang$ somewhere before. }
\begin{definition}[Translating  Judgements/Derivations]\label{def:firsttransl_judgments}
The translation of a type judgment for $\vasco$ into a type judgment for $\pilvl$ is parametric on the level function  $l:\mathcal{N} \to \mathbb{N}$, and is defined as: 
\begin{align*}
\sttoltj{\Gamma \st P}{l}{}& =\sttoltt{\Gamma}{l}{} \lt \sttoltp{P}{}
\\
\sttoltj{\Gamma, x:T \st x:T}{l}{} &= \sttoltt{\Gamma}{l}{}, x:\sttoltt{T}{l}{x}  \lt x:\sttoltt{T}{l}{x} 
	\end{align*}
	This translation induces  an inductive construction of the translation  of {type  derivations} in $\pilvl$ from type derivations in $\vasco$, denoted as:
\begin{mathpar}
	\sttoltj{{\scriptsize \vastype{Rule}}~\inferrule{\Upsilon_i \\ \forall i \in I}{\Gamma \st P}}{l}{} ~=~
		{\scriptsize\lvltype{Rule}}~\inferrule{\sttoltj{\Upsilon_i}{l}{}  \\ \forall i \in I}{ \sttoltt{\Gamma}{l}{} \lt \sttoltp{P}{}}
\end{mathpar} 
where $\Upsilon_i$ denotes a set of derivations used to prove $\Gamma \st P$. 

The translation, of which  \Cref{fig:encodWtoSi} gives an excerpt, relies on analyzing the last rule~\vastype{Rule} applied in the derivation $\Gamma \st P$ and the unfolding of the translation of judgements, mapping to a derivation $ \sttoltt{\Gamma}{l}{} \lt \sttoltp{P}{}$ in  $\pilvl$,  in which the last rule applied is~\lvltype{rule}.
\end{definition}

%  \joe{I am not sure I am using the correct word here}.
%	\daniele{The part below is confusing $\sttoltt{\cdot}{}{}$ is not a translation of judgements. I think you mean 'translation of contexts' or 'type assignments'}
%	
%	the parameters $\alpha, \beta, \gamma$ of undefined names within the translation of judgements $\sttoltt{ \cdot }{}{}$ are defined in  \Cref{f:lvltypeencodjudge} by unfolding of the translation of judgments. We omit the label $\vastype{Var}$ in the encoding of judgments.

\begin{figure*}
\small
$$
\begin{array}{|c|}
\hline
\textbf{Case: Input}\\
\hline
\begin{array}{lcl}
&&\\
%\vastype{Lin-In_1} && \lvltype{Lin-In_1}\\
\mprset{flushleft} 
\sttoltj{
\inferrule[\vastype{Lin-In_1} ]
{\Gamma_1' \st x: \lin \vasreci{T}{S}   \\
\Gamma_2 , x: S , y : T \st P } 
{ \Gamma_1'  \contcomp \Gamma_2 \st  \vasin{\lin}{x}{y}{P}  }
}{l}{} & = &
\mprset{flushleft}
\inferrule*[left= \lvltype{Lin-In_1}] {
					\sttoltt{\Gamma_1'}{l}{}  \lt x: \lvlint{l(x)}{\sttoltt{T}{l}{y}}{\sttoltt{S}{l}{z}} \\\\
					\sttoltt{\Gamma_2}{l}{} ,   y:\sttoltt{T}{l}{y} , z:\sttoltt{S}{l}{z}  \lt \sttoltp{P\substj{z}{x}}{}  \\ l(x) = l(z) }
				{
					\sttoltt{\Gamma_1'}{l}{} 	\contcomp
					\sttoltt{\Gamma_2}{l}{}
					\lt  x \inp {y , z}.\sttoltp{P\substj{z}{x}}{}
				}\\[.3cm] 
\text{where }\Gamma_1'=\Gamma_1,  x : \lin \vasreci{T}{S}&& 	\text{where } \sttoltt{\Gamma_1'}{l}{}=\sttoltt{\Gamma_1}{l}{} , x: \lvlint{l(x)}{\sttoltt{T}{l}{y}}{\sttoltt{S}{l}{z}}	\\[.2cm]
\hdashline
\mprset{flushleft}
&&\\
%\vastype{Lin-In_2}&&\lvltype{Lin-In_2}\\
\sttoltj{
    \inferrule[\vastype{Lin-In_2}]
     {  \Gamma_1' \st x: \recur{\wn T}  \\
       \Gamma_2', y : T \st P
		}
	{  \Gamma_1' \contcomp \Gamma_2'  \st  \vasin{\lin}{x}{y}{P} }
				}{l}{}  &=&
				\mprset{flushleft}
\mprset{flushleft}
\inferrule*[left=\lvltype{Lin-In_2}]
{
\sttoltt{\Gamma_1'}{l}{}  \lt x: \lvlunst{l(x)}{\sttoltt{T}{l}{y}} \\
\sttoltt{\Gamma_2'}{l}{} ,  z:\lvlnilT  \lt \sttoltp{P}{}  }
{  \sttoltt{\Gamma_1'}{l}{} 
\contcomp
\sttoltt{\Gamma_2'}{l}{}  
					\lt x\inp {y , z}.\sttoltp{P}{} 
	}\\[8mm]
\text{ where }\Gamma_1'= \Gamma_1 , x: \recur{\wn T} \text{ and }	\Gamma_2'= \Gamma_2 , x: \recur{\wn T} &&  \text{where }\sttoltt{ \Gamma_1'}{l}{}= \sttoltt{\Gamma_1}{l}{} , x: \lvlunst{l(x)}{\sttoltt{T}{l}{y}} \text{ and } \sttoltt{ \Gamma_2'}{l}{}= \sttoltt{\Gamma_2}{l}{} , x: \lvlunst{l(x)}{\sttoltt{T}{l}{y}}\\[.2cm]
%\qquad \quad \Gamma_2'= \Gamma_2 , x: \recur{\wn T}, 	&&  \qquad \quad \sttoltt{ \Gamma_2'}{l}{}= \sttoltt{\Gamma_2}{l}{} , x: \lvlunst{l(x)}{\sttoltt{T}{l}{y}}\\[.2cm]
\hdashline
&&\\
%\vastype{Un-In}&& \lvltype{Un-In_1}\\
\sttoltj{
\mprset{flushleft}
\inferrule[\vastype{Un-In}]
{ 
\Gamma' \st x: \recur{\wn T} \\
\Gamma' , y : T \st P }
{  \Gamma, x: \recur{\wn T} \st   \vasin{\un}{x}{y}{P}  }}{l}{}
			& = &
\mprset{flushleft}
\inferrule*[left= \lvltype{Un-In_1}]
{
\sttoltt{\Gamma'}{l}{} \lt x:\lvlunst{l(x)}{\sttoltt{T}{l}{y}} \\\\
\sttoltt{\Gamma'}{l}{} , y:\sttoltt{T}{l}{y},  z : \lvlnilT   \lt \sttoltp{P}{}  \\
 \forall b \in \outsubj{\sttoltp{P}{}}, \ \levelof{b} < \levelof{x}
}
{
\sttoltt{\Gamma}{l}{} ,  x:\lvlunst{l(x)}{\sttoltt{T}{l}{y}}  \lt \bang x \inp {y , z}.\sttoltp{P}{} 
}\\[5mm]
\text{ where } \Gamma'=\Gamma, x: \recur{\wn T}&&\text{ where }\sttoltt{\Gamma'}{l}{} =\sttoltt{\Gamma}{l}{} , x: \rlevel{l(x)}{\sttoltt{T}{l}{y}} \\[2mm]
\end{array}\\
\hline
\textbf{Case: Output}\\  
\hline
\\
\begin{array}{lcl}
%\vastype{Lin-Out} &&\lvltype{Lin-Out}\\
\sttoltj{
\mprset{flushleft}
\inferrule[\vastype{Lin-Out} ]
{
\Gamma_1' \st x: \lin \vassend{T}{S} \\
\Gamma_2' \st y: T \\
\Gamma_3, x:S \st P
}
{ \Gamma_1' \contcomp	\Gamma_2'\contcomp \Gamma_3 \st   \vasout{x}{y}{P}} 
}{l}{}
&=&
\mprset{flushleft}
\inferrule*[left=\lvltype{Lin-Out}]
{\inferrule{
\sttoltt{\Gamma_1'}{l}{} \lt x: \lvloutt{l(x)}{\sttoltt{T}{l}{y}}{\sttoltt{S}{l}{z}}\\\\
\sttoltt{\Gamma_2'}{l}{} \lt  y:\sttoltt{T}{l}{y} \\ 
\sttoltt{\Gamma_3}{l}{} , z : \sttoltt{S}{l}{z} \lt \sttoltp{P\substj{z}{x}}{}    \\ \levelof{x} = \levelof{z}
}
{
\sttoltt{(\Gamma_1'}{l}{}\contcomp \sttoltt{\Gamma_2'}{l}{} \contcomp \sttoltt{\Gamma_3}{l}{}), 
\lvlsplit{z}{\sttoltt{S}{l}{z}}  \lt \ov x \out { y,z }.(\sttoltp{P\substj{z}{x}}{}) 
}
}
{\sttoltt{\Gamma_1'}{l}{}\contcomp \sttoltt{\Gamma_2'}{l}{} \contcomp \sttoltt{\Gamma_3}{l}{} 
					 \lt \res{z} \ov x \out { y,z }.(\sttoltp{P\substj{z}{x}}{})
				}\\[7mm]
\text{ where } \Gamma_1'=\Gamma_1, x: \lin \vassend{T}{S} \text{ and } \Gamma'_2=\Gamma_2, y: T&&\text{where }\sttoltt{\Gamma_1'}{l}{}=  \sttoltt{\Gamma_1}{l}{}, x: \lvloutt{l(x)}{\sttoltt{T}{l}{y}}{\sttoltt{S}{l}{z}} \text{ and } \sttoltt{\Gamma_2'}{l}{}= \sttoltt{\Gamma_2}{l}{}, y:\sttoltt{T}{l}{y}\\[2mm] 
%\qquad \quad \Gamma'_2=\Gamma_2, y: T &&\qquad \quad \sttoltt{\Gamma_2'}{l}{}= \sttoltt{\Gamma_2}{l}{}, y:\sttoltt{T}{l}{y}\\[2mm]
%%%%%%
\hdashline
&&\\
%%%%%%
%\vastype{Un-Out} && \lvltype{Un-Out_1}\\
\sttoltj{
\mprset{flushleft}
\inferrule[\vastype{Un-Out}]
{
 \Gamma_1' \st x: \recur{\oc T}	\\
 \Gamma_2' \st y: T\\
\Gamma_3' \st P
}
{  \Gamma_1'\contcomp	\Gamma_2' \contcomp \Gamma_3' \st   \vasout{x}{y}{P}} 
				}{l}{} 
			&=&
\mprset{flushleft}
\inferrule*[left=\lvltype{Un-Out_1}]
{
\sttoltt{\Gamma_1'}{l}{} \lt x: \rlevel{l(x)}{\sttoltt{T}{l}{y}}\\\\
\sttoltt{\Gamma_2'}{l}{} \lt y:\sttoltt{T}{l}{y}\\ \\
\sttoltt{\Gamma_3'}{l}{}, z: \lvlnilT\lt \sttoltp{P}{} }
{
 \sttoltt{\Gamma_1'}{l}{} \contcomp	\sttoltt{\Gamma_2'}{l}{} 
					\contcomp \sttoltt{\Gamma_3'}{l}{}
					\lt \ov x \out { y,z }.\sttoltp{P}{} 
				}\\[7mm]
\text{ where } \Gamma_1'= \Gamma_1 , x: \recur{\oc T}, \  \Gamma_2'= \Gamma_2 , x: \recur{\oc T}, y: T  && \text{where } \sttoltt{\Gamma_1'}{l}{} =\sttoltt{\Gamma_1}{l}{} , x: \rlevel{l(x)}{\sttoltt{T}{l}{y}}, \  \sttoltt{\Gamma_2'}{l}{}=\sttoltt{\Gamma_2}{l}{}, x: \rlevel{l(x)}{\sttoltt{T}{l}{y}} , y:\sttoltt{T}{l}{y}\\
%\qquad \quad \Gamma_2'= \Gamma_2 , x: \recur{\oc T}, y: T   && \qquad \quad \sttoltt{\Gamma_2'}{l}{}=\sttoltt{\Gamma_2}{l}{}, x: \rlevel{l(x)}{\sttoltt{T}{l}{y}} , y:\sttoltt{T}{l}{y} \\
\qquad \quad \Gamma_3'= \Gamma_3 , x: \recur{\oc T} &&  \qquad \quad \sttoltt{\Gamma_3'}{l}{} =\sttoltt{\Gamma_3}{l}{} , x: \rlevel{l(x)}{\sttoltt{T}{l}{y}}\\[2mm]
\end{array}\\
\hline
\end{array}
$$
\jspace
\caption{From derivations in \vasco to derivations in \pilvl  (excerpt, cf.~\Cref{def:firsttransl_judgments})}\label{fig:encodWtoSi}
\jspace
\end{figure*}

%\begin{figure}[!t]
%%\scalebox{0.8}{
%\small
%\[
%\begin{array}{|rcl|}
%\hline
%&&\\
%\sttoltj{
%\inferrule[\vastype{Var}]{ \contun{\Gamma} }
%				{ \Gamma , x:T \st x:T }
%				}{l}{} 
%			& =&
%				
%				\inferrule[\lvltype{Var_1}]{ \contun{\sttoltt{\Gamma}{l}{}} }
%				{ \sttoltt{\Gamma}{l}{} , x:\sttoltt{T}{l}{x} \lt x:\sttoltt{T}{l}{x} }
%\\[7mm]
%\hdashline
%&&\\
% \sttoltj{\inferrule[\vastype{Nil} ]{ \contun{\Gamma}}
%				{\Gamma \st \nil}}{l}{} &=& \inferrule*[left= \lvltype{Nil}]{ \contun{\sttoltt{\Gamma}{l}{}}}
%				{ \sttoltt{\Gamma}{l}{} \lt \nil  }\\[7mm]
%%
%\hdashline
%&&\\
%\sttoltj{\inferrule[\vastype{Par}]{ \Gamma_1 \st P \\ \Gamma_2 \st Q }
%				{ \Gamma_1 \contcomp \Gamma_2 \st P \pp Q  }}{l}{}
%			&=&
%				\inferrule*[left=\lvltype{Par}]{ \sttoltt{\Gamma_1}{l}{} \lt \sttoltp{P}{}   \\ \sttoltt{\Gamma_2}{l}{} \lt \sttoltp{Q}{}  }
%				{ \sttoltt{\Gamma_1}{l}{} \contcomp \sttoltt{\Gamma_2}{l}{} \lt \sttoltp{P}{} \pp \sttoltp{Q}{}  }\\[7mm]
%%
%\hdashline
%&&\\
% \sttoltj{\inferrule[ \vastype{Res} ]{ \Gamma , x:T ,y: \dual{T} \st P }
%				{ \Gamma \st  \res {xy} P }}{l}{} 
%			&=&
%				\inferrule*[left=\lvltype{Res}]{ 
%					\sttoltt{\Gamma}{l}{} , \lvlsplit{z}{\sttoltt{T}{l}{x}} \lt \lvljoinsub{\sttoltp{P}{}}{z}{x,y}
%				}
%				{ \sttoltt{\Gamma}{l}{} \lt \res {z} \sttoltp{P\substj{z}{x}\substj{z}{y}}{} }\\[7mm]
%%
%\hline
%\end{array}
%\]
%%}
%\caption{From derivations in \vasco to derivations in \pilvl (\Cref{def:firsttransl_judgments}, Part 2)\label{fig:encodWtoSii}}
%\end{figure}

\subsection{Results}
In general, the translation of a $P \in \vaslang$ is not necessarily typable in $\pilvl$; this occurs when, e.g., $P$ is non-terminating.  
We focus on processes in $\vaslang$ that are typable in $\pilvl$, and therefore, are terminating.

\begin{notation}
	We write $\sttoltt{\Gamma}{l}{} \lt \sttoltp{P}{}$ if $\sttoltj{\Gamma \st P}{l}{}$ holds, for some~$l$.
	
	% \daniele{DN: we don't need this notation. It was established above already}
\end{notation}

Our translations are correct, in the following sense:

\begin{theorem}[Operational Completeness]\label{thm:op_completeness}
	Let $ P \in \lvllang$ such that  $\sttoltt{\Gamma}{l}{} \lt \sttoltp{P}{}$,  for some level function $l$. Then there exists $R \in \lvllang$ such that $P \vasred Q \implies\sttoltp{P}{}{} \lts{ \tau } \sttoltp{R}{}{}$  and $R \equiv Q$.
\end{theorem}

\begin{theorem}[Operational Soundness]\label{thm:op_soundness}
	Let $P \in \lvllang$ with  $\sttoltt{\Gamma}{l}{} \lt \sttoltp{P}{}$, for some level function $l$. If  $\sttoltp{P}{}{} \lts{ \tau } U$ Then there exists $R,Q \in \lvllang$ such that $P \vasred Q \land R \equiv Q \land U = \sttoltp{R}{}{} $. %$\sttoltp{P}{}{} \lts{ \tau } \sttoltp{R}{}{} \implies P \vasred Q $ and $ R \equiv Q $.
\end{theorem}
An immediate corollary of \Cref{thm:op_completeness} is that our translation preserves (non-)terminating behaviour, i.e.,  does not map non-terminating processes in $\vaslang$ into terminating processes in $\pilvl$.

\begin{corollary}
$\sttoltp{\cdot}{}$ preserves (non-)terminating behaviour.
\end{corollary}

The following result corroborates our informal intuitions about $\vaslang$ and $\lvllang$. It also precisely characterizes a class of terminating processes based on our correct translations  $\sttoltp{\cdot}{}$ and $\sttoltt{\cdot}{l}{}$.
\begin{theorem}
\label{t:wsubsets}
$\lvllang \subset \vaslang$.
\end{theorem}
\begin{proof}[Proof (Sketch)]
The inclusion $\lvllang \subseteq \vaslang$ is immediate by definition. 
To prove that the inclusion is strict, we consider a counterexample, i.e., a process $P$ typable in $\vasco$ but not typable in $\pilvl$.
Process $P_{\ref{ex:infinite}}$ from \Cref{ex:infinite} suffices for this purpose.
\end{proof}

\section{Propositions as Sessions}
\label{s:pas}
We now introduce $\pidill$, the process model induced by the Curry-Howard correspondence between linear types and session types  (\emph{propositions-as-sessions})~\cite{CairesP10}.
\pidill is a synchronous $\pi$-calculus extended with (binary) guarded choice and forwarding.

\begin{definition}[Processes and Types]
	Processes  in $\pidill$ are given by the grammar in \Cref{f:dillsyntandtype} (top). 	
	Types coincide with linear logic propositions, as given in the grammar in \Cref{f:dillsyntandtype} (bottom). 
\end{definition}

%	Process consist of input (linear and unrestricted), output, inaction, restriction and parallel composition. Processes may also be port forwarders $\dillforward{x}{y}$, provide choices of processes $P$ or $Q$ along a channel $x$ via $\dillchoice{x}{P}{Q}$ and finally either select the left $\dillselel{x}{P}$ or right $\dillseler{x}{P}$ choice along the channel $x$ and then continue as $P$.

\begin{definition}[Reduction in $\pidill$]
The reduction semantics of \pidill is defined in \Cref{f:dillcongred} (bottom), relying on structural congruence, the least congruence relation defined in \Cref{f:dillcongred} (top).  
\end{definition}

\begin{figure}[t!]
	\begin{mdframed}
	\[
	\begin{aligned}
		P,Q ::= &								&& \mbox{(Processes)}\\
				&  \dillout{x}{y}{P}		&& \mbox{(output)}
				&&  \dillin{x}{y}{P}					&& \mbox{(linear input)}\\
				&  \dillserv{x}{y}{P}			&& \mbox{(server)}
				&& P \pp Q  	    		  		&& \mbox{(composition)}\\
				& \res {x}P						&& \mbox{(restriction)}
				&& \nil	     					&& \mbox{(inaction)}\\
				& \dillchoice{x}{P}{Q}			&& \mbox{(branching)}
				&& \dillseler{x}{P}				&& \mbox{(select right)}\\
				& \dillselel{x}{P}				&& \mbox{(select left)}
				&& \dillforward{x}{y}			&& \mbox{(forwarding)}
	\end{aligned}
	\]
	\hrule
	\smallskip
	\[
	\begin{aligned}
		A,B ::= &								&& \mbox{(Types)}
		\\
				& \dillnilT						&& \mbox{(Termination)}  
				&& \dillunt{A}					&& \mbox{(Shared)}
				\\
				& \dillint{A}{B}				&& \mbox{(Receive)} 
				&& \dilloutt{A}{B}				&& \mbox{(Send)} 
				\\
				& \dillchoicet{A}{B}			&& \mbox{(Selection)} 
				&& \dillselet{A}{B}				&& \mbox{(Branching)} 
	\end{aligned}
	\]
	\end{mdframed}
	\jspace
	\caption{Processes and types of the session $\pi$-calculus $\pidill$}
	\label{f:dillsyntandtype}
	\jspace
\end{figure}

\begin{figure}[t!]
	\begin{mdframed}
		\[
		\begin{aligned}
			&\begin{aligned}
				P \pp \nil &\equiv P 
				&
				P \equiv_\alpha Q &\implies P \equiv Q
				&
				P \pp Q &\equiv Q \pp P
				\\
				\res{x} \nil &\equiv \nil
				&
				 P \pp (Q \pp R) &\equiv (P \pp Q) \pp R
				 &\!\!\res{x} \res{y} P &\equiv \res{y} \res{x} P
			\end{aligned}\\
			&\begin{aligned}
				x \not \in \fn{P} & \implies P \pp \res{x}Q \equiv \res{x} (P \pp Q)
				&
			\end{aligned}
		\end{aligned}
		\]
		\hrule
		\smallskip
		\[
			\begin{array}{rll}  
			\Did{R$\leftrightarrow$} & 
				\res{x} (\dillforward{x}{y}	\pp P)
				\dillred 
				 P\substj{y}{x} \quad \text{if } x \not = y
				\\[1mm] 
			\Did{RC} & 
				\dillout{x}{y}{P} \pp \dillin{x}{z}{Q}	
				\dillred 
				P \pp Q\substj{y}{z}
				\\[1mm] 
			\Did{R!} & 
				\dillin{x}{y}{P} \pp \dillserv{x}{z}{Q}
				\dillred
				P \pp \vaspara{Q\substj{y}{z}}  \dillserv{x}{z}{Q}
				\\[1mm]  
			\Did{RL}&
				\dillselel{x}{P} \pp  \dillchoice{x}{Q}{R}
				\dillred
				P \pp Q
			\\[1mm]
			\Did{RR}&
				\dillseler{x}{P}  \pp  \dillchoice{x}{Q}{R}\
				\dillred
				P \pp R
			\\[1mm]
			\Did{R$\pp$}&
				Q \dillred R 
				\implies 
				P \pp Q \dillred P \pp R
			\\[1mm]
			\Did{R$\nu$}&
				P \dillred Q 
				\implies 
				\res{x} P \dillred \res{x} Q
			\\[1mm]
			\Did{R$\equiv$}&
			P \equiv P' \land P' \dillred Q' \land Q' \equiv Q 
			\implies
			P \dillred Q
			\\[1mm]
			\end{array}
		\]
	\end{mdframed}
	\jspace
	\caption{Structural congruence and reductions for $\pidill$}
	\label{f:dillcongred}
	\jspace
\end{figure}

\begin{notation}[Process Abbreviations]
	We adopt the following abbreviations for bound outputs and replicated forwarders:
	\begin{align*}
	\dillbout{x}{z}{P} &=  \res{z}\dillout{x}{z}{P}
	\\
		\dillfwdbang{x}{y}  & = \dillserv{y}{z}{
		 \dillbout{x}{k}{	\dillforward{k}{z}}}
	\end{align*}

%	notation of $\dillbout{x}{z}{P}$ for expressing bound output $\res{z}\dillout{x}{z}{P}$. Also, we give an abbreviation the process that acts as a replicated forwarder as $ \dillfwdbang{x}{y}$ where 
%	$$\dillfwdbang{x}{y}  = \dillserv{y}{z}{
%		\res{k} \dillout{x}{k}{	\dillforward{k}{z}}
%	}$$
%$$\dillfwdbang{x}{y}  = \dillserv{y}{z}{
%		 \dillbout{x}{k}{	\dillforward{k}{z}}
%	}$$
%	Similarly we give an abbreviated type derivation as:
%	\begin{mathpar}
%		\inferrule* [ left = \dilltype{fwd^!} ]{ }
%		{\Gamma , x:A ; \cdot \dill \dillfwdbang{x}{y}:: y : \dillunt{A}}
%	\end{mathpar}
%	This abbreviates the derivation:
%	\begin{mathpar}
%		\inferrule* [ left = \dilltype{\dillunt{} R} ]{
%			\inferrule* [ left = \dilltype{copy} ]{
%				\inferrule* [ left = \dilltype{fwd} ]{
%				}{
%					\Gamma , x:A ; k:A \dill 
%					\dillforward{k}{z} :: y : A
%				}
%			}{
%				\Gamma , x:A ; \cdot \dill 
%				 \dillbout{x}{k}{	\dillforward{k}{z}} :: y : A
%			}
%		}{
%			\Gamma , x:A ; \cdot \dill \dillserv{y}{z}{
%				 \dillbout{x}{k}{	\dillforward{k}{z}}} :: y : \dillunt{A}
%		}
%	\end{mathpar}
\end{notation}

As usual, a type environment is a collection of type assignments $x : A $ where $x$ is a name and $A$ a type, the names being pairwise disjoint.  
The empty environment is denoted `$\cdot$'.
We consider  {\em unrestricted} environments (denoted $\Gamma,\Gamma'$) and  {\em linear} environments  (denoted as $\Delta,\Delta'$); while the former satisfy weakening and contraction, the latter do not.  

%, where type assignments in $\Gamma$ are propagated to all the premises whereas the type assignments in $\Delta$ are handled linearly, depending on the rule applied. 

We denote by $dom(\Gamma)$, the {\em domain of } $\Gamma$, the set of names whose type assignments are in $\Gamma$, i.e., $dom(\Gamma)=\{x \mid x:A \in \Gamma\}$.  Also,  $\Gamma(x)$ denotes the type of the name $x\in dom(\Gamma)$, i.e., $\Gamma(x)=A$, if $x:A\in \Gamma$. The domain of $\Delta$  and $\Delta(x)$ are similarly defined.

Typing judgments for $\pidill$ are of the form $\Gamma ; \Delta    \dill P \ :: x:A$.
Such a judgment is intuitively read as: ``$P$ provides protocol $A$ along $x$ by using the protocols described in the assignments in $\Gamma$ and $\Delta$''. The domains of  $\Gamma$, $\Delta$ and $x:A$ are pairwise disjoint.
	The corresponding type
	rules  are given  in \Cref{f:dilltyperule}.
	Each logical operator is represented by right and left rules: the former explains how to \emph{offer} a behavior (according to the operator's interpretation, cf. \Cref{f:dillsyntandtype} (bottom)); the latter explains how to \emph{make use} of a behavior typed with the operator.
In particular, the behavior of clients and servers is governed by four typing rules:
\dilltype{cut^!},
\dilltype{copy},
 \dilltype{\dillunt{} L}, and
 \dilltype{\dillunt{} R}.

The Curry-Howard correspondence connects the logical principle of cut elimination with process synchronization. As a result, we have  the fundamental property ensured by typing:

\begin{theorem}[Type Preservation]
If $\Gamma ; \Delta    \dill P \ :: x:A$
and $P \dillred Q$
then
$\Gamma ; \Delta    \dill Q \ :: x:A$.
\end{theorem}
		
The type system enforces also progress and termination. The latter property can be proven using logical relations~\cite{DBLP:journals/iandc/PerezCPT14}.

\begin{figure*}[t!]
	\begin{mdframed}
	\begin{mathpar}
		\inferrule* [ left = \dilltype{\dillnilT L} ]{  \Gamma ; \Delta \dill P :: T }
		{ \Gamma ; \Delta , x : \dillnilT  \dill P :: T }
		\and
		\inferrule* [ left = \dilltype{\dillnilT R} ]{   }
		{ \Gamma ; \cdot  \dill \nil ::   x : \dillnilT }
		\and
		\inferrule* [ left = \dilltype{fwd} ]{   }
		{ \Gamma ; x:A  \dill \dillforward{x}{y} ::   y : A }
		\and
		\inferrule* [ left = \dilltype{\dilloutt{}{} L} ]{ \Gamma ; \Delta , y:A , x:B \dill P :: T  }
		{ \Gamma ; \Delta , x :\dilloutt{A}{B} \dill \dillin{x}{y}{P} :: T  }
		\and
		\inferrule* [ left = \dilltype{\dilloutt{}{} R} ]{ \Gamma ; \Delta_1 \dill P :: y : A \\ \Gamma ; \Delta_2 \dill Q :: x : B  }
		{ \Gamma ; \Delta_1, \Delta_2  \dill \dillbout{x}{y}{(P \pp Q)} ::   x : \dilloutt{A}{B} }
		\and
		%
%		\inferrule* [ left = \dilltype{\dillint{}{} L} ]{  \Gamma ; \Delta_1 \dill P :: y:A \\ \Gamma ; \Delta_2 , x:B \dill Q :: T }
%		{ \Gamma ; \Delta_1 , \Delta_2  , x : \dillint{A}{B} \dill  \dillbout{x}{y}{(P \pp Q)} :: T }
%		%
%		\and
%		%
%		\inferrule* [ left = \dilltype{\dillint{}{} R} ]{  \Gamma ; \Delta , y:A \dill P :: x:B }
%		{ \Gamma ; \Delta  \dill \dillin{x}{y}{P} ::   x : \dillint{A}{B} }
%		%
%		\and 
		%
		\inferrule* [ left = \dilltype{cut} ]{ \Gamma ; \Delta_1 \dill P :: x: A \\ \Gamma ; \Delta_2 , x:A \dill Q ::T }
		{ \Gamma ; \Delta_1 ,  \Delta_2  \dill \res{x}(P \pp Q) ::  T }
		\and
		\inferrule* [ left = \dilltype{cut^!} ]{  \Gamma ; \cdot \dill P :: y: A \\ \Gamma , u:A ; \Delta  \dill Q ::T  }
		{ \Gamma ;   \Delta  \dill \res{u}(\dillserv{u}{y}{P} \pp Q) ::  T }
		\and
		\inferrule* [ left = \dilltype{copy} ]{ \Gamma , u:A ; \Delta , y:A \dill P  :: T  }
		{ \Gamma , u:A ; \Delta  \dill  \dillbout{u}{y}{P} ::  T }
		\and
		\inferrule* [ left = \dilltype{\dillunt{} L} ]{ \Gamma , u:A ; \Delta \dill P \substj{u}{x} :: T  }
		{ \Gamma ; \Delta , x: \dillunt{A} \dill P :: T }
		\and
		\inferrule* [ left = \dilltype{\dillunt{} R} ]{ \Gamma ; \cdot \dill Q :: y:A }
		{ \Gamma ; \cdot  \dill \dillserv{x}{y}{Q} ::   x : \dillunt{A} }
		\and
		\inferrule*[left=\dilltype{\dillselet{}{} L_1} ]{ \Gamma ; \Delta , x:A \dill 
		P :: T  }
		{ \Gamma ; \Delta , x:\dillselet{A}{B}  \dill \    \dillselel{x}{P} :: T }
		\and
		\inferrule* [ left = \dilltype{\dillchoicet{}{} L} ]{ \Gamma ; \Delta , x:A \dill P :: T \quad \Gamma ; \Delta , x:B \dill P :: T }
		{ \Gamma ; \Delta , x:\dillchoicet{A}{B} \dill   \dillchoice{x}{P}{Q} :: T }
		\and
		%
%		\inferrule* [ left = \dilltype{\dillselet{}{} R} ]{ \Gamma ; \Delta  \dill P ::x:A  \\ \Gamma ; \Delta  \dill P :: x:B }
%		{ \Gamma ; \Delta  \dill   \dillchoice{x}{P}{Q} :: x:\dillselet{A}{B} }
%		%
%		\and
		%
		\inferrule*[left=\dilltype{\dillchoicet{}{} R_2} ]{ \Gamma ; \Delta \dill 
		P ::  x:B  }
		{ \Gamma ; \Delta    \dill \    \dillseler{x}{P} :: x:\dillchoicet{A}{B} }
	\end{mathpar}
	\end{mdframed}
	\jspace
	\caption{Type rules for $\pidill$ (selection)}
	\label{f:dilltyperule}
	\jspace
\end{figure*}

\section{\dilllang: A Class of Terminating Processes}\label{s:dilllang}
We now study \dilllang, another class of terminating \vasco processes.
This class is induced by the Curry-Howard system given in \Cref{s:pas}, which leverages  
translations on processes and types/contexts, denoted $\sttodillp{\cdot}{}$ and $\sttodillt{\cdot}$, respectively.
Roughly, \dilllang is defined as follows:
	$$ \dilllang = \{ P \in \vasco \mid  \ \Gamma \vaslinunsplit \Delta \st P ~\land~ 
	 	 {\sttodillt{\Gamma}};   \sttodillt{\Delta}  \dill \sttodillp{{P}}  :: u: \sttodillt{\dual{S}} 	\} 
	$$
\Cref{d:dilllang} will give a formal definition.
In the following we define the translations $\sttodillp{\cdot}{}$ and $\sttodillt{\cdot}$, and establish their properties. 
Our main result is that 
$\dilllang \subset \lvllang$
but 
$\lvllang \not \subset \dilllang$
(\Cref{thm:wnotinl,thm:main_result}): there are terminating processes detected as such by the weight-based approach but not by the Curry-Howard correspondence.

%\daniele{(DN: I don't understand this paragraph.)}\joe{rewritten}
We require some auxiliary definitions.
%As $\pidill$ separates contexts into linear and unrestricted components, when encoding unrestricted types we wish to place then within the unrestricted context rather then the linear context. The type rule $ \dilltype{\dillunt{} L}$ allows for the moving of unrestricted types from the linear context to the unrestricted by removing the $\bang$ guarding it. We need a notation to perform this on encoded unrestricted types.
The following predicates say whether a session type contains client or server behaviors.

\begin{definition}\label{d:predsts}
Given a session type $T$, we define predicates $\server{{T}}$ and $\client{{T}}$ as follows:
	\[
	\begin{array}{rl@{\hspace{2cm}}rl}
		\server{\recur{\wn T}} &  = \ttrue 
		&
		\client{\recur{\oc T}} &  = \ttrue 
		\\
		\server{\nilT} &  = \ffalse 
		&
		\client{\nilT} &  = \ffalse 
		\\
		\server{q \ \vassend{S}{T}} &  = \server{T}
		&
		\client{q \ \vassend{S}{T}} &  = \client{T}
		\\
		\server{q \ \vasreci{S}{T}} &  = \server{T} 
		&
		\client{q \ \vasreci{S}{T}} &  = \client{T}
		\\
		\server{\recur{\oc T}} &  =  \server{T} 
		&
		\client{\recur{\wn T}} &  = \client{T}
	\end{array}
\]
These predicates extend to contexts $\Gamma$ as expected. This way, e.g., $\server{\Gamma}$ stands for $\bigwedge_{x\in dom(\Gamma)} \Gamma(x)$.
%is the conjunction of each names type in $\Gamma$. \daniele{(This sentence needs to be rewritten, the definition of $\server{\Gamma}$ is not clear. Perhaps we can give the formal definition $}
Also, we write $\server{\Gamma;P}$ to stand for $ \bigwedge_{x \in (\fn{P} \cap dom(\Gamma) )} \server{\Gamma(x)}$, returning $\ttrue$ when $(\fn{P} \cap dom(\Gamma)) = \emptyset$.
Analogous definitions for $\client{\cdot}$ , $\notserver{\cdot} $, and $ \notclient{\cdot} $ arise similarly. 
\end{definition}

%\begin{definition}[Core Context (in $\vaslang$)]
%	We define the core context of a process $P$ in $\vaslang$ typed with context $\Gamma$ as $\Gamma \st P$ then the core contest $\core{\Gamma}$ is the context such that $\core{\Gamma} \st P$ and $\forall x \in \dom{\core{\Gamma}}$ then $x \in \fn{P}$
%\end{definition}

This way, intuitively:
\begin{itemize}
	\item $\notserver{T}\land \notclient{T}$ means that $T$ is an always-linear behavior, i.e., it does not contain server and client actions.
	\item $\server{T} \land \notclient{T}$ means that $T$ contains some server behavior and that it does not contain client behaviors.
	\item $\notserver{T}\land \client{T}$ means that $T$ will at some point exhibit client behaviors and that it does not contain server behaviors.
\end{itemize}
Also, $\server{T} \land \client{T}$ means that $T$ contains both server and client actions; this combination, however, is excluded by typing. 
 
\begin{example}%[Examples of $\server{\cdot}$ and $\client{\cdot}$]
	We further illustrate \Cref{d:predsts} by example:
	\begin{center}
		\begin{tabular}{ cc|c|c } 
		% \hline
		 &$T$ & $\server{T}$ & $\client{T}$  \\
		 \hline 
		 1& $\recur{ \oc ( \lin \  \vassend{(\recur{\oc S})}{(\lin \ \vasreci{(\recur{\wn R})}{(\recur{\wn T_0})})})}$ & \ttrue & \ttrue   \\
		 %\hline
		 2& $\lin \ \vassend{(\recur{\oc S})}{(\lin \ \vasreci{(\recur{\wn R})}{(\recur{\wn T_0})})}$ & \ttrue  & \ffalse  \\
		 %\hline
		 3& $\recur{\oc (\lin \ \vassend{(\recur{\oc S})}{(\lin \ \vasreci{(\recur{\wn R})}{\nilT})})}$& \ffalse & \ttrue    \\
		 %\hline
		 4& $\lin \ \vassend{(\recur{\oc S})}{(\lin \ \vasreci{(\recur{\wn R})}{\nilT})}$ & \ffalse & \ffalse \\
		 %\hline
		\end{tabular}
	\end{center} 
	Both $(1)$ and $(2)$ return \ttrue  for $\server{T}$ because of their final behavior (i.e., `$\recur{\wn T_0}$'), whereas $(3)$ and $(4)$ return \ffalse, because their final  behavior is $\nilT$. 
	Both $(1)$ and $(3)$ return \ttrue for $\client{T}$ as their initial type behavior (i.e., `$\recur{\oc T'}$') is that of a client, whereas $(2)$ and~$(4)$ return \ffalse as they do not contain any client behavior.
\end{example}

\subsection{The Typed Translation}

\begin{definition}[Translating Processes]
The translation $\sttodillp{\cdot}: \vasco \to \pidill$ is given in~\Cref{f:encproc}. %,  is defined for well-typed processes in $\vaslang$, and relies on channel types and  on its behaviour as server or client.
\end{definition}

\begin{figure*}[!t]
	\begin{mdframed}
	\[
	\begin{aligned}
		\sttodillp{\vasout{x}{y}{P}} & = 
			\begin{cases}
				 \dillbout{x}{z}{( \dillforward{y}{z}  \pp  \sttodillp{{P}}   )} 
					& \text{If $ x: \lin \oc (T).S \land \neg \contun{T}\land \notserver{T}$.}
					\\
				 \dillbout{x}{z}{( \dillfwdbang{y}{z}  \pp  \sttodillp{{P}}   )}
					& \text{If $ x: \lin \oc (T).S \land \contun{T} \land \notserver{T}$.}
					\\
				%\res{z} \dillout{x}{z}{ \dillselel{z}{
				%	\res{w} \dillout{z}{w}{
				%		( \dillfwdbang{v}{w}  \pp  \sttodillp{{P}}   )
				%	}
				%}
				%} \\
				\dillbout{x}{z}{ 
					\dillbout{z}{w}{
						( \dillforward{y}{w}  \pp  \sttodillp{{P}}   )
					}
				} 
					&
				  \text{If $ x: \recur{\oc T} \land \neg \contun{T} \land  \notserver{T} \land \client{T}$}
					\\
				\dillbout{x}{z}{ 
					\dillbout{z}{w}{
						( \dillfwdbang{y}{w}  \pp  \sttodillp{{P}}   )
					}
				}
					&
			 \text{If $ x: \recur{\oc T} \land \contun{T} \land \notserver{T} \land \client{T}$}
			 \\
			 				\dillbout{x}{z}{ \dillselel{z}{
					\dillbout{z}{w}{
						( \dillfwdbang{y}{w}  \pp  \sttodillp{{P}}   )
					}
				}
				}
					&
					  \text{If $ x: \recur{\oc T} \land \contun{T}\land \notserver{T} \land \notclient{T} $}
			\\
			 				\dillbout{x}{z}{ \dillselel{z}{
					\dillbout{z}{w}{
						( \dillforward{y}{w}  \pp  \sttodillp{{P}}   )
					}
				}
				}
					&
					  \text{If $ x: \recur{\oc T} \land \neg \contun{T}\land \notserver{T} \land \notclient{T} $}
			\end{cases}\\
		\sttodillp{\res{xy}\vasout{z}{y}{P}  } & = 
			\begin{cases}
				\dillbout{z}{x}{\dillseler{x}{\sttodillp{P}}}  & \text{If $z:\recur{\oc T} \land \notserver{T} \land \notclient{T} $.}\\
				\dillbout{z}{x}{\sttodillp{P}} & \text{If $z:\recur{\oc T} \land \server{T} \land \notclient{T} $.}
			\end{cases}
		\\
		\sttodillp{\res{xy}\vasout{z}{x}{(P \pp Q)}} & = 
			\dillbout{z}{y}{(\sttodillp{P} \pp \sttodillp{Q})}
			\quad \text{If $z:  \lin \vassend{T}{S} \wedge z \not \in \fn{P} \land y \not \in \fn{Q}$}
		\\
		\sttodillp{\vasin{\lin}{x}{y}{P}} & = 
				\dillin{x}{y}{\sttodillp{P}} \quad \text{If $x: \lin \vasreci{T}{S}$}\\
		\sttodillp{\vasin{\un}{x}{{y}}{P}} & = 
			\begin{cases}
				\dillserv{x}{z}{\sttodillp{P\substj{z}{y}}} 
				& \text{If $ x: \recur{\wn T} \land\server{T}\land \notclient{T} $.} \\
				\dillserv{x}{z}{ \dillin{z}{y}{ \sttodillp{P} }
				} 
				& \text{If $ x: \recur{\wn T} \land\notserver{T}\land \client{T} $.} \\
				\dillserv{x}{z}{  
				\dillchoice{z}{	\dillin{z}{y}{ \sttodillp{P} }	}{	\sttodillp{P\substj{z}{y}} }
				} 
				& \text{If $ x: \recur{\wn T} \land\notserver{T}\land \notclient{T} $.}
			\end{cases}\\ 
			\sttodillp{\res{x y}(P \pp Q)} & = 
					\res{x}( \sttodillp{P} \pp \sttodillp{Q}\substj{x}{y}) \quad \text{If $ y \not \in \fn{P} \land x \not \in \fn{Q}$}\\
		\sttodillp{P \pp Q} & = \res{w} ( \sttodillp{P} \pp \sttodillp{Q}   ) \quad \text{With $ w $ fresh}
		\\
				\sttodillp{\nil} & = \nil 
	\end{aligned}
	\]
		\end{mdframed}
		\jspace
\caption{Translating processes in \vasco into $\pidill$ \label{f:encproc}}
\jspace
\end{figure*}

The translation of processes relies on type information; in particular, the translation of outputs and unrestricted inputs depends on whether the overall behavior of channels exhibits server or client behaviors (cf. \Cref{d:predsts}). 
In translating  outputs, we check whether the output is free or bound.  
The translation of free outputs is further influenced by whether the sender is associated with a linear connection or acts as a client connected to a server. There are 5 cases to consider, and the translated processes are designed to preserve typability.
Similar conditions apply to the translation of bound outputs.

%\daniele{TO POLISH:}
\begin{remark}
To ensure typability of the translated process, we explain some of the choices in~\Cref{f:encproc}:
\begin{enumerate}
\item In a free output $\vasout{x}{z}{P}$ the value $z$   cannot have a server behavior. In~\Cref{f:encproc}, this is ensured using the predicate $\notserver{T}$.
\item In an unrestricted bound output $\res{xy}\vasout{z}{y}{P} $,  the value $y$ cannot have a client behavior. In~\Cref{f:encproc}, this is ensured using the predicate $\notclient{T}$. %  (cf.~\Cref{f:encproc}).
\end{enumerate}

We illustrate what we mean by ``client behavior'' above. Consider the process $P=\res{xy}( \res{wv}\vasout{x}{v}{\vasin{\un}{w}{a}{\nil}}\pp \vasin{\un}{y}{c}{\vasout{c}{b}{\nil}})$.
In $P$, the output action on $x$ is an unrestricted bound output, whose object $v$ has a client behavior: after one reduction, an output on $v$ will be ready to invoke the server on $w$. 
 Notice that $P \in \vaslang$, as $P$ is typable with $b:\nilT\st P$ and $x:\recur{\oc (\recur{\oc \nilT})}, y:\recur{\wn (\recur{\oc \nilT})}, w:\recur{\wn \nilT}$ and $ v:\recur{\oc \nilT}$. % (see also~\Cref{ex:derivation}).

%\daniele{DN: The next paragraphs rely on  tricky arguments because we are using *our* translation to explain *why our* translation cannot have client sent by unrestricted bound outputs. I tried to add some extra comments about the fact that our translation is reasonable in the sense that it 'respects behaviour' from source to target language. But this is quite an subjective concept.}

We want a typable  translation of the judgement $\sttodillj{\Gamma \st P}_u$. Consider the partial translation of $P$, i.e., $\sttodillp{P}=\res{x}(\sttodillp{P_1}\pp \sttodillp{Q_1}\substj{x}{y})$, 
%(\daniele{DN: translating restriction in $\vasco$ to a restriction in $\pidill$ is self-explainable, and also the fact that the translation acts homomorphically inside the process}),
 where we use the abbreviations
\begin{itemize}
\item  $P_1= \res{wv}\vasout{x}{v}{\vasin{\un}{w}{a}{\nil}}$, and 
\item  $Q_1= \vasin{\un}{y}{c}{\vasout{c}{b}{\nil}}$.
\end{itemize}
 Suppose we can apply \dilltype{cut}, then there are derivations $\Pi_1$ and $\Pi_2$ such that 
\begin{mathpar} 
\small
\inferrule
{
\inferrule{\Pi_1}{b:\dillnilT; \cdot \dill  \sttodillp{Q_1}\substj{x}{y}:: y: !( \dillint{\dillnilT}{\dillnilT})}\quad
\inferrule{\Pi_2}{b:\dillnilT; x:!( \dillint{\dillnilT}{\dillnilT})\dill \sttodillp{P_1} :: u:T}
}
{b:\dillnilT;\cdot \dill \res{x}( \sttodillp{P_1}\pp \sttodillp{Q_1}\substj{x}{y})::u:T}
\end{mathpar}

Consider the partial translation $\sttodillp{\res{wv}\vasout{x}{v}{\vasin{\un}{w}{a}{\nil}}}= \dillbout{x}{w}{(\dillserv{w}{a'}{P_1'})}$ for some $P_1'$ that we will leave opaque for now.
% (\daniele{DN: Note that bound output in \vasco is translated to bound output in \pidill; similarly to unrestricted input.}). 
Notice, however, that the following derivation is not possible: to type $\dillserv{w}{a'}{P_1'}$ we would need $u=w$ to apply $\dilltype{\dillunt{}R}$ (above the application of $\dilltype{copy}$), but $w$ already occurs in the context and this contradicts the domain restriction of $\dill$ judgements.
\begin{mathpar}
\inferrule*[left=\dilltype{\dillunt{}L}]
{
\inferrule*[left=\dilltype{copy}]
{b:\dillnilT , x:( \dillint{\dillnilT}{\dillnilT});w:( \dillint{\dillnilT}{\dillnilT}) \not \dill !w(a).0::u:T}
{b:\dillnilT , x:( \dillint{\dillnilT}{\dillnilT});\cdot  \dill \dillserv{w}{a'}{P_1'}::u:T}}
{b:\dillnilT ; x:!( \dillint{\dillnilT}{\dillnilT})\dill \dillbout{x}{w}{(\dillserv{w}{a'}{P_1'})}::u:T}
\end{mathpar}
A similar argument and example can be used to justify the first item of this remark.
\end{remark}

While the translation of linear inputs is straightforward, in translating unrestricted inputs we check whether the synchronization concerns a bound or free output. 
When the unrestricted input cannot discern the client or server behavior from the type, it offers both behaviors using a branching construct; the synchronizing party (i.e. the translation of output, free or bound) then determines the desired behavior using a corresponding selection construct.
%Notice how the encoding of parallel composition is always defined, this is as typing is what disallows this encoding.

\begin{example} \label{ex:trans_proc} Consider $P=\res{xy}(  \vasin{\un}{ x}{z}{\nil}\pp  \vasout{y}{w}{\nil})$, a $\vasco$ process that implements a simple  server-client communication.  As in \Cref{ex:infinite}, one can verify that  
$x:\recur{\wn \nilT}, y:\recur{\oc \nilT}$, $w:\nilT$ and $z:\nilT$, which  entail $w:\nilT \st P$. Since $y\notin \fn{\vasin{\un}{x}{z}{\nil}}$ and $x\notin \fn{\vasout{y}{w}{\nil}}$, the translation of $P$ is as:
\[
\begin{aligned}
\sttodillp{P}= \res{x}(\sttodillp{\vasin{\un}{x}{z}{\nil}}\pp \sttodillp{\vasout{x}{w}{\nil}})
\end{aligned}
\]
Note that $\notclient{\nilT}\wedge \notserver{\nilT}\wedge \contun{\nilT}$ holds (cf.~\Cref{def:pred_cont} and \Cref{d:predsts}). Thus,
\[
\begin{aligned}
\sttodillp{\vasin{\un}{x}{z}{\nil}} &=\dillserv{x}{v}{  
				\dillchoice{v}{	\dillin{v}{z}{ \nil }	}{\nil }
				} \\
\sttodillp{\vasout{x}{w}{\nil}} &= \dillbout{x}{z}{ \dillselel{z}{
					\dillbout{z}{v}{
						( \dillfwdbang{w}{v}  \pp  \nil   )
					}
				}}
\end{aligned}
\]
\end{example}

\begin{definition}\label{d:nobang}
Given a session type/linear logic proposition $A$, we write $\dillnotunsingular{{A}}$ to denote $A$ without 
 top-level occurrences of `\,$\bang$\,', i.e., $\dillnotunsingular{\dillunt{A}} =  A$ and is  the identity function otherwise.
%\[
%	\begin{aligned}
%		\dillnotunsingular{\dillunt{T}} &  =  T &
%		\dillnotunsingular{\dillnilT} &  = \dillnilT &
%		\dillnotunsingular{\dilloutt{S}{T}} &  =\dilloutt{S}{T}  \\
%		\dillnotunsingular{\dillint{S}{T}} &  = \dillint{S}{T} &
%		\dillnotunsingular{\dillchoicet{S}{T}} &  = \dillchoicet{S}{T} &
%		\dillnotunsingular{\dillselet{S}{T}} &  = \dillselet{S}{T} 
%	\end{aligned}
%\]	
\end{definition}

\begin{definition}[Translating Types/Contexts]
The translation $\sttodillt{\cdot }$ from session types in $\vasco$ to logic propositions in $\pidill$ is given in~\Cref{f:enctype}. % for all defined cases, otherwise undefined.
The translation of types extends to contexts as expected;
we shall write $\dillnotun{\sttodillt{\Gamma}}$ to stand for  $\dillnotunsingular{\sttodillt{\Gamma}}$.
\end{definition}

\begin{figure}[!t]
\begin{mdframed}
\[
	\begin{aligned}
		\sttodillt{\nilT} &  = \dillunt{\dillnilT} \\
		\sttodillt{\lin \vassend{S}{T}} &  =  \dillint{\sttodillt{S}}{\sttodillt{T}}\\
		\sttodillt{\lin \vasreci{S}{T}} &  = \dilloutt{\sttodillt{S}}{\sttodillt{T}} \\
		\sttodillt{\recur{\wn T}} &  = 
			\begin{cases}
				\dillunt{\sttodillt{T}} & \text{If $\server{T}\land \notclient{T} $.}\\
				\dillunt{(\dilloutt{\sttodillt{T}}{\dillnilT})} & \text{If $\notserver{T}\land \client{T} $.}\\
				\dillunt{(\dillchoicet{(\dilloutt{\sttodillt{T}}{\dillnilT})}{\sttodillt{T}})} & \text{If $\notserver{T}\land \notclient{T} $.}  
				%\\
				%\text{Undefined} & \text{If $\server{T}\land \client{T} $.}
			\end{cases}\\
		\sttodillt{\recur{\oc T}} &  =  
			\begin{cases}
				\dillunt{\sttodillt{\dual{T}}} & \text{If $\server{T}\land \notclient{T} $.} \\
				\dillunt{({\dillint{\sttodillt{T}}{\dillnilT}})} & \text{If $\notserver{T}\land \client{T} $.}\\
				\dillunt{(\dillselet{(\dillint{\sttodillt{T}}{\dillnilT})}{\sttodillt{\dual{T}}})} & \text{If $\notserver{T}\land \notclient{T} $.}
				%\\
				%\text{Undefined}  & \text{If $\server{T}\land \client{T} $.}
			\end{cases}\\
	\end{aligned}
\]
\end{mdframed}
\jspace
\caption{Translating session types into logical propositions \label{f:enctype}}
\jspace
\end{figure}

The translation of $\nilT$ and linear input/output types is standard.
As for client and servers, the translation of types follows the translation of processes. 
When the type of the client or server exhibits a server behavior, the type is encoded into an unrestricted type.
Notice that a client type $\recur{\oc T}$ is translated into its dual behavior $\dillunt{\sttodillt{\dual{T}}} $, but a server is not. 
This has to do with the left/right interpretation of judgments in \pidill: servers always occur on the right-hand side; to provide a dual behavior, the client should itself be dual.

\begin{example}[Cont.~\Cref{ex:trans_proc}]\label{ex:transl_type}
Consider the type assignments 
$x:\recur{\wn \nilT}, y:\recur{\oc \nilT}$, $w:\nilT$ and $z:\nilT$. Since  $\notclient{\nilT}\wedge \notserver{\nilT}$, the translation in~\Cref{f:enctype} gives: 
\begin{itemize} 
\item $x:\sttodillt{\recur{\wn\nilT}}=\dillunt{(\dillchoicet{(\dilloutt{\sttodillt{\nilT}}{\dillnilT})}{\sttodillt{
\nilT}})}=\dillunt{(\dillchoicet{(\dilloutt{\dillunt{\dillnilT}}{\dillnilT})}{\dillnilT})}$;
\item $y:\sttodillt{\recur{\oc \nilT}}=\dillunt{(\dillselet{(\dillint{\sttodillt{\nilT}}{\dillnilT})}{\sttodillt{\dual{\nilT}}})}=\dillunt{(\dillselet{(\dillint{\dillunt{\dillnilT}}{\dillnilT})}{~\dillnilT})}$
\end{itemize}
The translations to $z:\dillnilT$ and $w:\dillnilT$ are trivial.
\end{example}

Armed with the translations of processes and types given in~\Cref{f:encproc} and  \Cref{f:enctype}, 
we are now ready to translate a judgment $\Gamma, \Delta \st P$ into $\dillnotun{\sttodillt{\Gamma}{}{}} ;  \sttodillt{\Delta}{}{} \dill \sttodillp{P}{} :u:: A$, for some name $u$.
This translation requires that $\contun{\Gamma}$ and $\neg \contun{\Delta}$, i.e., $\Gamma$ is unrestricted and $\Delta$ is `not' unrestricted; this is the abbreviation $\Gamma \vaslinunsplit \Delta$ (\Cref{note:sep}).

The following auxiliary notion relates contexts that may differ in exactly one assignment:
	\begin{definition}%[\joe{CONTEXT THING NEED TO NAME}] 
		Given contexts  $\Gamma , \Gamma'$, we write  $\Gamma \dillcontrel{z}{T} \Gamma' $ if $(\Gamma = \Gamma' \land z \not \in \dom{\Gamma}) \lor (\Gamma = \Gamma' , z:T) $ for some type $T$. 
		%The relation  $\dillcontrel{z}{T}$ is defined for linear contexts $\Delta,\Delta'$ in a similar way.
	\end{definition}

\begin{definition}[Translating Judgements]
\label{d:transjudgdill}
Given a judgment $\Gamma \vaslinunsplit \Delta \st P$ 
and a name $u$, 	its translation   $\sttodillj{\Gamma \vaslinunsplit \Delta \st P}_{u}{}$ is defined as 
	$$ \dillnotun{\sttodillt{\Gamma'}{}{}} ;  \sttodillt{\Delta'}{}{} \dill \sttodillp{P}{} :: u: A$$ 
	where $\Gamma'$, $\Delta'$, and $A$ are subject to one of the following conditions:
	\begin{itemize}
		\item  $A= \sttodillt{\dual{T}}{}{}$ when $\{ u: T \} \subset \Gamma , \Delta$ with $ (\Gamma \dillcontrel{u}{T} \Gamma') \land (\Delta \dillcontrel{u}{T} \Delta')$;~or
		\item  $A= \dillnilT $ when $ u \not\in \dom{\Gamma , \Delta}$ with $ (\Gamma = \Gamma') \land (\Delta = \Delta')$.
	\end{itemize}
	
\Cref{d:secondtablep1} defines the translation by induction on $P$, assuming that the contexts satisfy the appropriate requirements, i.e., $\contun{\Gamma,\Gamma'} \land \neg \contun{\Delta , \Delta_1,\Delta' }$ and $A,\Gamma'$ and $\Delta'$ are as one of the cases above.  
	
%	The table uses the following abbreviations:
%			\begin{align*}
%				\abconditionO{u}{\Gamma}{\Delta}{\Gamma'}{\Delta'}{A}  
%				& = 
%				u \not \in \dom{\Gamma,\Delta} \land A = \dillnilT \land \Gamma = \Gamma' \land \Delta = \Delta'
%				\\
%				\abconditionT{u}{R}{\Gamma}{\Delta}{\Gamma'}{\Delta'}{A} 
%				& = \{u : R\} \subset \Gamma , \Delta \land A =  \sttodillt{\dual{R}} \land \Gamma \dillcontrel{u}{R} \Gamma' \land \Delta \dillcontrel{u}{R} \Delta'
%				%\\
%				%\abconditionTstar{u}{R}{\Gamma}{\Delta}{\Gamma'}{\Delta'}{A}{P}{z} 
%				%& = \abconditionT{u}{R}{\Gamma}{\Delta}{\Gamma'}{\Delta'}{A} \land u \not \in \fn{P}
%			\end{align*}
	\end{definition}

\begin{table*}[t!]
\small
	\begin{tabular}[c]{| p{.2cm} | c | l |}
		\hline
%		\multicolumn{3}{|c|}{
%			$
%			\begin{aligned}
%				\abconditionO{u}{\Gamma}{\Delta}{\Gamma'}{\Delta'}{A}  %= 				\conditionO{u}{\Gamma}{\Delta}{\Gamma'}{\Delta'}{A} 
%				& = 
%				u \not \in \dom{\Gamma,\Delta} \land A = \dillnilT \land \Gamma = \Gamma' \land \Delta = \Delta'
%				\\
%				\abconditionT{u}{R}{\Gamma}{\Delta}{\Gamma'}{\Delta'}{A} %=				\conditionT{u}{R}{\Gamma}{\Delta}{\Gamma'}{\Delta'}{A} 
%				& = \{u : R\} \subset \Gamma , \Delta \land A =  \sttodillt{\dual{R}} \land \Gamma \dillcontrel{u}{R} \Gamma' \land \Delta \dillcontrel{u}{R} \Delta'
%				\\
%				\abconditionTstar{u}{R}{\Gamma}{\Delta}{\Gamma'}{\Delta'}{A}{P}{z} %=				\conditionTstar{u}{R}{\Gamma}{\Delta}{\Gamma'}{\Delta'}{A}{P}{z} 
%				& = \abconditionT{u}{R}{\Gamma}{\Delta}{\Gamma'}{\Delta'}{A} \land u \not \in \fn{P}
%				\\
%			\end{aligned}
%			$
%		}
%		\\
%		\hline \hline
		& $\Gamma \vaslinunsplit \Delta \st P$ 
		& $\dillnotun{\sttodillt{\Gamma}} ;  \sttodillt{\Delta} \dill \sttodillp{P} :: u : A$ 
		\\
		\hline \hline
		 	1&
			$
				{ \Gamma  \vaslinunsplit \cdot  \st \nil  }
			$
			& 
			$
				{ \dillnotun{\sttodillt{\Gamma}} ;  \cdot \dill \nil :: u : \dillnilT }
			$
			\\
		\hline
		%&& \\
			2&
			$
				\Gamma \vaslinunsplit \Delta_1 , \Delta \st P \pp Q 
			$
			& 
			$
			\begin{array}{ll}
				%\begin{aligned}
					\dillnotun{\sttodillt{\Gamma'}} ; \sttodillt{\Delta_1} , \sttodillt{\Delta'}   \dill \res{w} ( \sttodillp{P} \pp \sttodillp{Q}   )  ::  u: A
				%\end{aligned}
				&
				\text{if } w \not \in \dom{\Gamma, \Delta_1, \Delta}  \land u \notin \fn{P}
				%\land 
%				(
%					\abconditionO{u}{\Gamma}{\Delta}{\Gamma'}{\Delta'}{A}
%					\lor \abconditionT{u}{R}{\Gamma}{\Delta}{\Gamma'}{\Delta'}{A}  
%				)
			\end{array}
			$
			\\
		%&& \\
		\hline
		%&& \\
			3&
			$
			\begin{array}{l}
				\Gamma \vaslinunsplit  \Delta_1 ,  \Delta \st 
				  \res {zv:V} (P \pp Q)
			\end{array}
			$
			& 
			$
			\begin{array}{l}
				\text{If } 
				(u\notin \fn{P})
				\land
				( (\neg \contun{V}) \lor 
				( \contun{V} \land v \not \in \fn{P} \land z \not \in \fn{Q} )) \
%				\land (
%					\abconditionO{u}{\Gamma}{\Delta}{\Gamma'}{\Delta'}{A}
%					\lor 
%					\abconditionT{u}{R}{\Gamma}{\Delta}{\Gamma'}{\Delta'}{A}
%				) 
				\\

				%\begin{aligned}
					\dillnotun{\sttodillt{\Gamma'}} ; \sttodillt{\Delta_1} , \sttodillt{\Delta'}  \dill \res{z}( \sttodillp{P} \pp \sttodillp{Q}\substj{z}{v} ) ::  u: A 
				%\end{aligned}
				%&
				%\text{if } \neg \contun{V} \land
				%(
				%	\abconditionO{u}{\Gamma}{\Delta}{\Gamma'}{\Delta'}{A}
				%	\lor 
				%	\abconditionTstar{u}{R}{\Gamma}{\Delta}{\Gamma'}{\Delta'}{A}{P}{u} 
				%)
			\end{array}
			$
			\\[2mm]
		%&& \\
		\hline
		%&& \\
		%	4&
		%	$
		%	\begin{array}{l}
		%		 \Gamma  \vaslinunsplit \Delta_1 , \Delta \st \\
		%		 \quad  \res {zv:V} (P \pp Q)
		%	\end{array}
		%	$
		%	& 
		%	$
		%	\begin{array}{ll}
				%\begin{aligned}
		%			\dillnotun{\sttodillt{\Gamma'}} ; \sttodillt{\Delta_1} , \sttodillt{\Delta'}  \dill \res{z}( \sttodillp{P} \pp \sttodillp{Q}\substj{z}{v} ) ::  u: A
				%\end{aligned}
		%		&
		%		\begin{array}{l}
		%			\text{if } \contun{V} \land v \not \in \fn{P} \land
		%			\\ z \not \in \fn{Q} \land
		%			(
		%				\abconditionO{u}{\Gamma}{\Delta}{\Gamma'}{\Delta'}{A}
		%				\lor 
		%				\abconditionTstar{u}{T}{\Gamma}{\Delta}{\Gamma'}{\Delta'}{A}{P}{u} 
		%			)
		%		\end{array}
		%	\end{array}
		%	$
		%	\\
		%&& \\
		%\hline
		%&& \\
			4&
			$
			\begin{array}{ll}
				\Gamma , x: \recur{\wn T} \vaslinunsplit \cdot \st  
				  \vasin{\un}{x}{y}{P}
			\end{array}
			$
			&
			$
			\begin{array}{l}
				\text{If } u = x \land x \not \in \fn{P} \text{ and one of the following holds:}
				\\
				\begin{tabular}{ll}
					%\multicolumn{2}{l}{
						$
						%\begin{array}{l}
							%\text{if } \notserver{T}\land \notclient{T} 
							%\\
							\dillnotun{\sttodillt{\Gamma}} ; \cdot  \dill   \dillserv{x}{w}{  
							\dillchoice{w}{	\dillin{w}{y}{ \sttodillp{P} }	}{	\sttodillp{P\substj{w}{y}} }} 
							 :: u: \dillunt{(\dillselet{(\dillint{\sttodillt{T}}{\dillnilT})}{\sttodillt{\dual{T}}})} $
							 &
							 $\text{if } \notserver{T}\land \notclient{T}
						%\end{array}
						$
					%}
					\smallskip
					\\
					\hdashline 
					\noalign{\smallskip}
					$\dillnotun{\sttodillt{\Gamma}} ; \cdot  \dill   \dillserv{x}{w}{  
						{ \sttodillp{P\substj{w}{y}} }} ::  u: \dillunt{\sttodillt{\dual{T}}} $ 
					& 
					\text{if } $\server{T} \land \notclient{T} $
					\smallskip
					\\
					\hdashline 
					\noalign{\smallskip}
					$\dillnotun{\sttodillt{\Gamma}} ; \cdot  \dill   \dillserv{x}{w}{  
						\dillin{w}{y}{ \sttodillp{P} }	} ::  u: \dillunt{(\dillint{\sttodillt{T}}{\dillnilT})}$  
					& 
					\text{if } $\notserver{T} \land \client{T}$
					\smallskip
				\end{tabular}
			\end{array}
			$
			\\
		\hline
		%&& \\
			5&
			$
			\begin{array}{l}
				\Gamma \vaslinunsplit  x: \lin \vasreci{T}{S} , \Delta \st 
				   \vasin{\lin}{x}{y}{P}
			\end{array}	
			$
			& 
			\begin{tabular}{ll} 
			$\dillnotun{\sttodillt{\Gamma}} ; \sttodillt{\Delta} \dill \dillin{x}{y}{\sttodillp{P}}  :: u:\dillint{\sttodillt{T}}{\sttodillt{\dual{S}}}$ & \text{if $u = x$}
			\smallskip
			\\
			\hdashline 
			\noalign{\smallskip}
			$\dillnotun{\sttodillt{\Gamma'}} ; \sttodillt{\Delta'} , x:\dilloutt{\sttodillt{T}}{\sttodillt{S}}  \dill \dillin{x}{y}{\sttodillp{P}} :: u: A$ & 
			\text{otherwise} %\text{with $\abconditionO{u}{\Gamma}{\Delta}{\Gamma'}{\Delta'}{A} \lor \abconditionT{u}{R}{\Gamma}{\Delta}{\Gamma'}{\Delta'}{A} $}
			\smallskip
			\end{tabular}
%			\begin{minipage}{9cm}
%				\vskip 4pt
%				\begin{enumerate}
%				\item If $u = x$
%
%					$
%					\begin{aligned}
%						& \dillnotun{\sttodillt{\Gamma}} ; \sttodillt{\Delta} \dill \dillin{x}{y}{\sttodillp{P}}  :: x:\dillint{\sttodillt{T}}{\sttodillt{\dual{S}}}
%					\end{aligned}
%					$	
%
%				\item $\conditionO{u}{\Gamma}{\Delta}{\Gamma'}{\Delta'}{A} \lor \conditionT{u}{R}{\Gamma}{\Delta}{\Gamma'}{\Delta'}{A} $
%
%				$
%				\begin{aligned}
%					& \dillnotun{\sttodillt{\Gamma'}} ; \sttodillt{\Delta'} , x:\dilloutt{\sttodillt{T}}{\sttodillt{S}}  \dill \dillin{x}{y}{\sttodillp{P}} :: u: A
%				\end{aligned}
%				$
%				
%				\end{enumerate}
%				\vskip 4pt
%			\end{minipage}
			\\
		\hline
		%&& \\
			6&
			$
			\begin{aligned}
				& \Gamma , z: \recur{\oc T} \vaslinunsplit \Delta  \st
				  \res{xy}\vasout{z}{y}{P}
			\end{aligned}	
			$
			&
			$
			\begin{array}{l}
				\text{If } (\neg \contun{T}) \lor (y \not \in \fn{P} \land \contun{T}) \text{ and:}
				\\
				\begin{tabular}{ll}
					%\multicolumn{2}{l}{
							$\dillnotun{\sttodillt{\Gamma'}} , z:\dillselet{\dillint{\sttodillt{T}}{\sttodillt{\dillnilT}}}{\sttodillt{\dual{T}}}  ; \sttodillt{\Delta'}  \dill 
							\dillbout{z}{x}{
							\dillseler{x}{\sttodillp{P}}} ::  u: A  
						%\end{array}
						$
						&
						$
						%\begin{array}{l}
							\text{if } \notserver{T} \land \notclient{T}  
							%\land (\abconditionO{u}{\Gamma}{\Delta}{\Gamma'}{\Delta'}{A} \lor \abconditionT{u}{R}{\Gamma}{\Delta}{\Gamma'}{\Delta'}{A} ) 
							$
					%}
					\smallskip
					\\
					\hdashline 
					\noalign{\smallskip}
					$ \dillnotun{\sttodillt{\Gamma'}} , z:\sttodillt{\dual{T}}  ; \sttodillt{\Delta'}  \dill  \dillbout{z}{x}{
						\sttodillp{P}} :: u: A $
					& 
					\text{if } $\server{T} \land \notclient{T} 
%					\land
%					(\abconditionO{u}{\Gamma}{\Delta}{\Gamma'}{\Delta'}{A} \lor \abconditionT{u}{R}{\Gamma}{\Delta}{\Gamma'}{\Delta'}{A} ) 
					$
					\smallskip
				\end{tabular}
			\end{array}
			$
			\\
		\hline
		%&& \\
			7&
			$
			\begin{array}{l}
				\Gamma \vaslinunsplit z: \lin \vassend{T}{S} , \Delta_1, \Delta \st \\
				\quad   \res{xy}\vasout{z}{x}{(P \pp Q)}
			\end{array}
			$
			&
			\begin{tabular}{ll} 
				\multicolumn{2}{l}{
						$\text{If } (\neg \contun{T}) \lor (x \not \in \fn{P} \cup \fn{Q} \land \contun{T}) \text{ and:}$
				}\\
				\multicolumn{2}{l}{
					$
					\begin{array}{ll}
						%\text{If }z = u \land u,z \not \in \fn{P} \land y \not \in \fn{Q}
						%\\
						\dillnotun{\sttodillt{\Gamma}} ; \sttodillt{\Delta_1} , \sttodillt{\Delta} \dill
						\dillbout{z}{y}{(\sttodillp{P} \pp \sttodillp{Q})} ::   z : \dilloutt{\sttodillt{T}}{\sttodillt{\dual{S}}} & \qquad \quad  \text{if }z = u \land u,z \not \in \fn{P} \land y \not \in \fn{Q}
					\end{array}
					$
				}
				\smallskip
				\\
				\hdashline 
				\noalign{\smallskip}
				\multicolumn{2}{l}{
					$ 
					\begin{array}{ll}
						%\text{If }{z \not = u} \land {u,z \not \in \fn{P}}\land y \not \in \fn{Q} 
						%\land ( u \notin \fn{P}) 
%						\land
%						(
%							\abconditionO{u}{\Gamma}{\Delta}{\Gamma'}{\Delta'}{A}
%							\lor 
%							\abconditionT{u}{R}{\Gamma}{\Delta}{\Gamma'}{\Delta'}{A}
%						)
						%\\
						\dillnotun{\sttodillt{\Gamma'}} ; \sttodillt{\Delta_1} , \sttodillt{\Delta'}  , z : \dillint{\sttodillt{T}}{\sttodillt{S}} \dill 
						\dillbout{z}{y}{(\sttodillp{P} \pp \sttodillp{Q})} :: u: A  & \text{if }{z \not = u} \land {u,z \not \in \fn{P}}\land y \not \in \fn{Q} 
					\end{array}
					$
				}
				\smallskip
			\end{tabular}
			\\
		\hline
					8&
			$
			\begin{array}{l}
			\smallskip
				\Gamma  \vaslinunsplit v:T, x: \lin \oc (T).S , \Delta 
				 \st \ov x\out v.P
			\end{array}
			$
			&
			\begin{tabular}{ll}
				%If $u = x \land \neg \contun{T} \land \notserver{T}$
				%\\
				$ \dillnotun{\sttodillt{\Gamma}} ;  v: \sttodillt{T}, \sttodillt{\Delta}   \dill \dillbout{x}{y}{(\dillforward{v}{y}  \pp  \sttodillp{P})} ::   u : \dilloutt{\sttodillt{T}}{\sttodillt{\dual{S}}} $ & if $u = x \land \neg \contun{T} \land \notserver{T}$
				\smallskip
				\\
				\hdashline 
				\noalign{\smallskip}
%				If  $\neg \contun{T} \land \notserver{T} $
%				\land
%				(
%					\abconditionO{u}{\Gamma}{\Delta}{\Gamma'}{\Delta'}{A}
%					\lor 
%					\abconditionT{u}{R}{\Gamma}{\Delta}{\Gamma'}{\Delta'}{A} 
%				)  
				%\\
				$ \dillnotun{\sttodillt{\Gamma'}} ;  v: \sttodillt{T},  \sttodillt{\Delta'}  , x : \dillint{\sttodillt{T}}{\sttodillt{S}}  \dill \dillbout{x}{y}{(\dillforward{v}{y} \pp P)} ::  u: \sttodillt{\dual{R}}  $ &	if  $\neg \contun{T} \land \notserver{T} $
				\smallskip
			\end{tabular}
			\\
		\hline
		%&& \\
			9&
			$
			\begin{array}{l}
				\Gamma,  v:T \vaslinunsplit x: \lin \oc (T).S , \Delta 
				 \st \ov x\out v.P 
			\end{array}
			$
			&
			\begin{tabular}{ll}
%				If $u = x \land \contun{T} \land \notserver{T}$
%				\\
				$  \dillnotun{\sttodillt{\Gamma}} , v: \sttodillt{T} ;  \sttodillt{\Delta} \dill \dillbout{x}{y}{(\dillfwdbang{v}{y}  \pp  \sttodillp{P})} ::   x : \dilloutt{\sttodillt{T}}{\sttodillt{\dual{S}}}$ & if $u = x \land \contun{T} \land \notserver{T}$

				\smallskip
				\\
				\hdashline 
				\noalign{\smallskip}
%				If  $\contun{T} \land \notserver{T}
%				 \land
%				(
%					\abconditionO{u}{\Gamma}{\Delta}{\Gamma'}{\Delta'}{A}
%					\lor 
%					\abconditionT{u}{R}{\Gamma}{\Delta}{\Gamma'}{\Delta'}{A} 
%				)  
%				$
%				\\
				$ \dillnotun{\sttodillt{\Gamma'}},  v:T  ; \sttodillt{\Delta'}  , x : \dillint{\sttodillt{T}}{\sttodillt{S}}  \dill \dillbout{x}{y}{(\dillfwdbang{v}{y} \pp P)} ::  u: A $ & if  $\contun{T} \land \notserver{T}$
				\smallskip
			\end{tabular}
			\\
		\hline
		%&& \\
			10&
			$
			\begin{aligned}
				& \Gamma , x: \recur{\oc T} \vaslinunsplit v:T , \Delta 
				 \st \ov x\out v.P
			\end{aligned}
			$
			&
			\begin{tabular}{l} 
				If $\notserver{T} \land \notclient{T}  \land \neg \contun{T} \land z \not \in \fn{P}$  $
%				\land
%				(
%					\abconditionO{u}{\Gamma}{\Delta}{\Gamma'}{\Delta'}{A}
%					\lor 
%					\abconditionT{u}{R}{\Gamma}{\Delta}{\Gamma'}{\Delta'}{A} 
%				) 
				$
				\\
				$ \dillnotun{\sttodillt{\Gamma'}} , x:\dillselet{\dillint{\sttodillt{T}}{\dillnilT}}{\sttodillt{\dual{T}}} ;  v:\sttodillt{T} , \sttodillt{\Delta'} \dill \dillbout{x}{z}{ \dillselel{z}{ \dillbout{z}{w}{( \dillforward{v}{w}  \pp  \sttodillp{{P}} )}}} :: u: A$
				\smallskip
				\\
				\hdashline 
				\noalign{\smallskip}
				If $\notserver{T} \land \client{T}  \land \neg \contun{T} \land z \not \in \fn{P}$ $
%				\land
%				(
%					\abconditionO{u}{\Gamma}{\Delta}{\Gamma'}{\Delta'}{A}
%					\lor 
%					\abconditionT{u}{R}{\Gamma}{\Delta}{\Gamma'}{\Delta'}{A} 
%				) 
				$
				\\
				$ \dillnotun{\sttodillt{\Gamma'}} , x:\dillint{\sttodillt{T}}{\dillnilT} ;  v:\sttodillt{T} , \sttodillt{\Delta'}   \dill \dillbout{x}{z}{ \dillbout{z}{w}{ ( \dillforward{v}{w}  \pp  \sttodillp{{P}} )} } ::  u: A $
				\smallskip
			\end{tabular}
			\\
		\hline
		%&& \\
			11&
			$\begin{aligned}
				& \Gamma , x: \recur{\oc T}, v:T\vaslinunsplit \Delta 
				 \st \ov x\out v.P
			\end{aligned}
			$
			& 
			$
			\begin{array}{l}
			\text{If } \notserver{T} \land \notclient{T} \land \contun{T}\land z \notin \fn{P}\\ 
			\dillnotun{\sttodillt{\Gamma'}} , x:\dillselet{\dillint{\sttodillt{T}}{\dillnilT}}{\sttodillt{\dual{T}}} ; v:\sttodillt{T} ; \sttodillt{\Delta'} \dill \dillbout{x}{z}{ \dillselel{z}{ \dillbout{z}{w}{( \dillfwdbang{v}{w}  \pp  \sttodillp{{P}} )}}} :: u: A		\smallskip
			\\
			\hdashline
				\text{If  } \notserver{T} \land \client{T} \land \contun{T} \land z \not \in \fn{P}
%				 \land 
%				(
%					\abconditionO{u}{\Gamma}{\Delta}{\Gamma'}{\Delta'}{A}
%					\lor
%					\abconditionT{u}{R}{\Gamma}{\Delta}{\Gamma'}{\Delta'}{A} 
%				)
				\\
				\begin{aligned}
					& \dillnotun{\sttodillt{\Gamma'}} , x:\dillint{\sttodillt{T}}{\dillnilT} , v:\sttodillt{T};   \sttodillt{\Delta'}  \dill \dillbout{x}{z}{ 
					\dillbout{z}{w}{
						( \dillfwdbang{v}{w}  \pp  \sttodillp{{P}}   )
						}
					} :: u: A
				\end{aligned}
			\end{array}
			$
			\\
	%	&& \\
		\hline
	\end{tabular}
	\caption{From judgments in \vasco to judgments in \pidill (\Cref{d:transjudgdill}).  \label{d:secondtablep1}}
	\jspace
	\end{table*}

Using this translation of judgments, a translation of derivations can be defined exactly as in \Cref{def:firsttransl_judgments}.

We discuss some entries in \Cref{d:secondtablep1} with the following example.

\begin{example}[Cont.~\Cref{ex:trans_proc}] 
%\daniele{DN: replace $\dillnilT$ with $\dillunt{\dillnilT}$ in some places below. Check carefully.}

The translation of judgments defined in~\Cref{d:secondtablep1} relies on the typability in the source language (\vaslang) to determine the exact conditions to the typability of the translated process in~$\pidill$. First, the type derivation of $w:\nilT\st \res{xy}(  \vasin{\un}{ x}{z}{\nil}\pp  \vasout{y}{w}{\nil})$ is essential to build a type derivation for the translated judgment (if one exists):
\begin{mathpar}
\small
\inferrule
{
\inferrule{
\inferrule{\Gamma\st x:\recur{\wn \nilT} \\
 \Gamma, z:\nilT \st \nil}{\Gamma\st \vasin{\un}{x}{z}{\nil}}
 \\
\inferrule{\Gamma\st y:\recur{\oc \nilT}\\\\
\Gamma \st w:\nilT\\
\Gamma\st \nil}
{\Gamma\st \vasout{y}{w}{\nil}}
 }
 {\Gamma \st \vasin{\un}{ x}{z}{\nil}\pp  \vasout{y}{w}{\nil}}
 }
{w:\nilT \st \res{xy} (\vasin{\un}{ x}{z}{\nil}\pp  \vasout{y}{w}{\nil})}
\end{mathpar}
where $\Gamma= x:\recur{\wn \nilT}, y:\recur{\oc \nilT}, w:\nilT$

Second, the translation  $\sttodillj{w:\nilT\st \res{xy}(  \vasin{\un}{ x}{z}{\nil}\pp  \vasout{y}{w}{\nil})}_u$
 corresponds to entry 3 of \Cref{d:secondtablep1}, 
\begin{mathpar}
\inferrule
{}
{w:{\dillnilT};\cdot \dill \res{x}(\sttodillp{\vasin{\un}{ x}{z}{\nil}}\pp \sttodillp{\vasout{x}{w}{\nil}})::u:A}
\end{mathpar}
and we have previously observed that the side conditions hold. Since $u$ is not in the type context, we have $A=\dillnilT$. Now we proceed to build a type derivation for the translated judgment by applying rule $\dilltype{cut}$:
\begin{mathpar}
\small
\inferrule%*[left=\dilltype{cut}]
{
\inferrule{\Pi_1}
{w:{\dillnilT}; \cdot\dill \sttodillp{\vasin{\un}{ x}{z}{\nil}}::x:  \dillunt{B}
} \quad
\inferrule{\Pi_2
\\\\
w:{\dillnilT},y :B; \cdot \dill  \sttodillp{\vasout{y}{w}{\nil}})::u:\dillnilT}{w:{\dillnilT}; x :\dillunt{B} \dill  \sttodillp{\vasout{x}{w}{\nil}})::u:\dillnilT}
}
{w:{\dillnilT};\cdot \dill \res{x}(\sttodillp{\vasin{\un}{ x}{z}{\nil}}\pp \sttodillp{\vasout{x}{w}{\nil}})::u:\dillnilT}
\end{mathpar}
and we will show that there exist derivations $\Pi_1$ and $\Pi_2$  such that the derivation holds.  We recall the translations $\sttodillp{\vasin{\un}{x}{z}{0}}$ and $\sttodillp{\vasout{x}{w}{
\nil}}$ in \Cref{ex:trans_proc}.  
%Next we proceed by analyzing the translation of the  premises. 

Third, the left premise is the translation $\sttodillj{\Gamma\st \vasin{\un}{x}{z}{\nil}}_x$  and corresponds to entry (4) of the~\Cref{d:secondtablep1}. We  recall \Cref{ex:transl_type} for $\dillnotun{\sttodillt{\Gamma}}=w:\dillnilT, x:  (\dillchoicet{(\dilloutt{\dillunt{\dillnilT}}{\dillnilT})}{\dillnilT}) , y: (\dillselet{(\dillint{\dillunt{\dillnilT}}{\dillnilT})}{\dillnilT})$ and use the abbreviation $B=(\dillselet{(\dillint{\dillunt{\dillnilT}}{\dillnilT})}{\dillnilT})$ and strengthening of $\dillnotun{\sttodillt{\Gamma}}$. The derivation $\Pi_1$ is as follows:
\begin{mathpar}
\small
\inferrule*[left=\dilltype{\dillunt{R}}]{
\inferrule*[left=\dilltype{\dillselet{}{}R}]{
\inferrule{
\inferrule{
 w:{\dillnilT}; \cdot  \dill \nil :: v: \dillnilT
}
{ w:{\dillnilT}; z:\dillnilT \dill \nil :: v: \dillunt{\dillnilT}}
}
{ w:{\dillnilT}; \cdot\dill \dillin{v}{z}{\nil} ::v: \dillint{\dillunt{\dillnilT}}{\dillnilT}} \\
 w:{\dillnilT}; \cdot\dill \nil :: v: \dillnilT
}
{ w:{\dillnilT}; \cdot \dill   \dillchoice{v}{	\dillin{v}{z}{ \nil }	}{\nil }
							 :: v: \dillselet{\dillint{(\dillunt{\dillnilT}}{\dillnilT})}{\dillnilT} }
}
{ w:{\dillnilT}; \cdot\dill \dillserv{x}{v}{  \dillchoice{v}{	\dillin{v}{z}{ \nil }}{\nil }} ::x: \dillunt{(\dillselet{(\dillint{\dillunt{\dillnilT}}{\dillnilT})}{\dillnilT})}
}
\end{mathpar}
%.

Fourth, the right premise is the translation $\sttodillj{\Gamma\st \vasout{y}{w}{0}\substj{x}{y}}_u$ and  corresponds to the entry (11a) of the~\Cref{d:secondtablep1}. The derivation $\Pi_2$ is as follows:

\begin{mathpar}
\small
%\inferrule{
\inferrule*[left=\dilltype{copy}]{
\inferrule{
\inferrule{
\dillnotun{\sttodillt{\Gamma'}} ; \cdot  \dill  \dillfwdbang{w}{v} :: v :\dillunt{\dillnilT}\quad 
\dillnotun{\sttodillt{\Gamma'}} ; z:\dillnilT    \dill   \nil  :: u :\dillnilT}
{\dillnotun{\sttodillt{\Gamma'}} ; z:(\dillint{\dillunt{\dillnilT}}{\dillnilT})   \dill \dillbout{z}{v}{( \dillfwdbang{w}{v}  \pp  \nil )} :: u :\dillnilT
}}{
\dillnotun{\sttodillt{\Gamma'}} ; z:\dillselet{(\dillint{\dillunt{\dillnilT}}{\dillnilT})}{\dillnilT}   \dill \dillselel{z}{ \dillbout{z}{v}{( \dillfwdbang{w}{v}  \pp  \nil )}} :: u: \dillnilT
}}
{\dillnotun{\sttodillt{\Gamma'}}  ; \cdot  \dill \dillbout{y}{z}{ \dillselel{z}{ \dillbout{z}{v}{( \dillfwdbang{w}{v}  \pp  \nil )}}} :: u: \dillnilT}
%{\dillnotun{\sttodillt{\Gamma'}}; x :\dillunt{\dillselet{(\dillint{\dillnilT}{\dillnilT})}{\dillnilT}} \dill \dillbout{x}{z}{ \dillselel{z}{ \dillbout{z}{v}{( \dillfwdbang{w}{v}  \pp  \nil )}}}::u:\dillnilT}
\end{mathpar}
where we strengthen $\dillnotun{\sttodillt{\Gamma}}$ to  $\dillnotun{\sttodillt{\Gamma'}}= y: (\dillselet{(\dillint{\dillunt{\dillnilT}}{\dillnilT})}{\dillnilT}), w:{\dillnilT}$.
\end{example}

We have the following property, which holds by definition of the entries of \Cref{d:secondtablep1}:

\begin{theorem}[Type preservation]\label{thm:second_type_pres}
	If $\Gamma \vaslinunsplit \Delta \st P$ then $ \dillnotun{\sttodillt{\Gamma'}{}{}} ;  \sttodillt{\Delta'}{}{} \dill \sttodillp{P}{} :: u: A$ is well-typed in \pidill, with  $A, \Gamma'$ and $\Delta'$ as in \Cref{d:transjudgdill}.
\end{theorem}

Notice that the translations of typable \vasco processes are not necessarily typable in \pidill.  
We shall concentrate on  processes in $\vaslang$ that are typable in $\pidill$:
\begin{notation}
	We write $\dillnotun{\sttodillt{\Gamma'}};   \sttodillt{\Delta'}  \dill \sttodillp{{P}}  :: u: \sttodillt{\dual{S}} $ whenever $\sttodillj{\Gamma ,\Delta  \st P}_u$  holds, with $\Gamma \dillcontrel{u}{S} \Gamma'$ and $\Delta \dillcontrel{u}{S} \Delta'$. 
\end{notation}

We can finally define $\dilllang$:

\begin{definition}[$\dilllang$]
\label{d:dilllang}
Let $u$ be a name.	We define:
	\begin{align*}
	 \dilllang = \{ P \in \vasco \mid &  \ \Gamma \vaslinunsplit \Delta \st P \land 	\Gamma \dillcontrel{u}{S} \Gamma' 
	 \land  \Delta \dillcontrel{u}{S} \Delta'  \\	& \land \dillnotun{\sttodillt{\Gamma'}};   \sttodillt{\Delta'}  \dill \sttodillp{{P}}  :: u: \sttodillt{\dual{S}} 	\} 
	\end{align*}
%		$$ \dilllang = \{ P \in \vasco \mid  \exists \Gamma , \Delta, \Gamma', \Delta', u, S \ s.t. \ \Gamma \vaslinunsplit \Delta \st P \land 
%	\Gamma \dillcontrel{u}{S} \Gamma' \land  \Delta \dillcontrel{u}{S} \Delta'
%\land 
%	% \sttodillj{\Gamma , \Delta \st P}_{u} =
%	\dillnotun{\sttodillt{\Gamma'}};   \sttodillt{\Delta'}  \dill \sttodillp{{P}}  :: u: \sttodillt{\dual{S}} 	\} 
%	$$
where contexts and types mentioned are existentially quantified.
\end{definition}

\subsection{Results}
%We state our main results:
\begin{theorem}[$\dilllang \subset \lvllang$]\label{thm:main_result}
	Let $P\in \vaslang$ such that $\Gamma \st P$, for some context $\Gamma$. If there exists $  u$  such that  $\sttodillj{\Gamma \st P}_u $ holds, then there exists $l$ such that $ \sttoltj{\Gamma \st P}{l}{}$ holds.
	%if $p \in \dilllang$ with $\sttodillj{\Gamma \st P}$ then $p \in \lvllang$ 
	%and $\exists \ l \ s.t. \ \sttoltj{\Gamma , \Delta \st p}{l}{}$
\end{theorem}

The proof of \Cref{thm:main_result} is 
	by induction on the structure of $P$. 
	We exploit a number of invariant properties for the type systems for \pidill and \pilvl, including:
	\begin{itemize}
		\item In $\pidill$, judgments for typed processes never exhibit servers on the left-hand side.
		\item In \pilvl, levels for types that do not exhibit server behavior can be decreased at will.
		\item  The type system of \pidill ensures that the name on the right-hand side is not guarded by servers.
	\end{itemize}
%\end{proof}

\begin{theorem}[$\lvllang \not \subset \dilllang$]
\label{thm:wnotinl}
	 $\exists P \in \lvllang$  with $\Gamma \st P$ and $\sttoltj{\Gamma \st P}{l}{} $ for some $l$  such that $ \nexists  \ z \ s.t. \ \sttodillj{\Gamma \st P}_z $.
\end{theorem}

To prove \Cref{thm:wnotinl}, it suffices to consider the \vasco process
	\[
		P = 	
		\vasres{xy}{(
			\vaspara{ 
				\vasin{\lin}{x}{z}{\vasin{\un}{z}{w}{\nil}} 
			}
			{
				\vasres{st}{
					\vasout{y}{s}{(
						\vaspara{
							\vasres{uv}{(\vasout{t}{u}{\nil})}
						}
						{
							\nil
						}
					)}
				}
			}
		)}
	\]
Clearly, $P$ is terminating:
	\[
		\begin{aligned}
			P &\redd 
					\vasres{st}{(
						\vaspara{ 
							\vasin{\un}{s}{w}{\nil}
						}
						{
							\vasres{uv}{(\vasout{t}{u}{\nil})}
						}
					)}
				 \redd 
					\vasres{st}{
						\vasin{\un}{s}{w}{\nil}
					}
		\end{aligned}
	\]
Process $P$ can be typed so as to establish $P \in  \vaslang$. 
Also, there is a level function that makes its translation into \pilvl typable. Hence, $P \in \lvllang$. 
However, its translation into \pidill is not typable, so $P \not\in \dilllang$.

%\begin{remark}
%	There are many interesting counter examples of $\lvllang \not \subset \dilllang$
%
%	\[
%	\begin{aligned}
%		1 & \vasin{\un}{x}{y}{\vasin{\un}{x}{z}{P}} &
%		2 & \vasin{\un}{x}{y}{\vasin{\un}{z}{w}{P}} \\
%		3 & (\vaspara{\vasin{\un}{x}{y}{P}}{\vasin{\un}{x}{z}{Q}}) &
%		4 & (\vaspara{\vasin{\un}{x}{y}{P}}{\vasin{\un}{z}{w}{Q}}) \\
%		5 & \vasin{\lin}{x}{y}{\vasin{\un}{y}{z}{P}} &
%		6 & \vasres{yz}{(\vasout{x}{{y}}{(
%			\vaspara{\vasres{wv}{\vasout{z}{w}{P}}}{Q_x}
%			)})}  \\
%		7 & \vasres{yz}{\vasout{x}{y}{\vasin{\un}{z}{w}{P}}} &
%		8 & \vasin{\un}{x}{y}{ \vasres{wv}{\vasout{y}{w}{P}}} \\
%	\end{aligned}
%	\]
%	The proof \Cref{thm:wnotinl} can be seen as $5$ and $6$ cut together as these illistarte more interesting case of examples of the differences between the two type systems.
%\end{remark}

\section{Closing Remarks}
\label{s:close}
We presented a comparative study of type systems for concurrent processes in the $\pi$-calculus, from the unifying perspective of termination and session types. 
To our knowledge, this is the first study of its kind.
Even by focusing on only three different type systems, we were confronted with technical challenges connected with the intrinsic differences between them.
The typed process model \vasco~\cite{V12}, focused on session-based concurrency, admits a rather broad class of processes, exploiting a clear distinction between linear and unrestricted resources, implemented via context splitting.  
The typed process model \pilvl combines features from type systems that target the termination property~\cite{DS06} and type systems for sessions. 
Finally, the typed process model \pidill~\cite{CairesP10} rests upon a firm logical foundation, and its control of clients and servers is directly inherited from the logical principles of the exponential $!A$.
Notice that \pidill is unique among type systems for the $\pi$-calculus in that it ensures  protocol fidelity, deadlock-freedom, confluence, and strong normalization/termination for typed processes.

The main take-away message is that the Curry-Howard correspondence is strictly weaker than weight-based approaches for enforcing the termination property. 
Hence, the  control of server/client interactions that is elegantly enabled by the copying semantics of  $!A$ turns out to be rather implicit when contrasted to weight-based techniques.
Interestingly,  Dardha and P\'{e}rez arrived to a similar conclusion in their comparative study of type systems focused on the deadlock-freedom property~\cite{DardhaP15,DBLP:journals/jlap/DardhaP22}: type systems based on the Curry-Howard correspondence can detect strictly less deadlock-free processes than other, more sophisticated type systems. 
Notice that the study in~\cite{DardhaP15,DBLP:journals/jlap/DardhaP22} considers only finite processes, without input-guarded replication (so all process are terminating).

Immediate items for current and future work include incorporating other type systems into our formal comparisons. 
The type systems by Sangiorgi~\cite{DBLP:journals/mscs/Sangiorgi06} and by Yoshida et al.~\cite{DBLP:conf/lics/YoshidaBH01} are very appealing candidates. %; we conjecture that it induces a class of processes closely related to \dilllang. 
Also, Deng and Sangiorgi proposed several type systems for termination.  
Here we considered only the simplest variant, which induces the class \lvllang and is already different from \dilllang; it would be interesting to consider the other variants.

\section*{Acknowledgments}
%This work has been supported by 
We  acknowledge the support of 
the Dutch Research Council (NWO) under project No. 016.Vidi.189.046 (Unifying Correctness for Communicating Software).
Daniele Nantes-Sobrinho has been supported by the EPSRC Fellowship `VeTSpec: Verified Trustworthy Software Specification' (EP/R034567/1) and Edital DPI/DPG n.~03/2020.

We are grateful to 
Davide Sangiorgi,
Nobuko Yoshida,  and 
the anonymous reviewers for useful suggestions and remarks.

%\begin{acks}
%To Robert, for the bagels and explaining CMYK and color spaces.
%\end{acks}

%\bibliographystyle{ACM-Reference-Format}
\bibliographystyle{alpha}
\bibliography{my_biblio.bib}

\ifarxiv
\appendix

\newpage 
\onecolumn
\input{appendix}
\else 
\fi

\end{document}